\documentclass[12pt, draftclsnofoot, onecolumn]{IEEEtran}

\usepackage{amsmath}
\usepackage{amssymb}
\usepackage{amsfonts}
\usepackage{amsthm}
\usepackage{color}
\usepackage{bm}
\usepackage{graphicx}
\usepackage{float}
\usepackage{soul}
\usepackage{subfigure}
\usepackage{cite}
\usepackage{algorithm,algpseudocode}
\usepackage{setspace}
\usepackage{bbm}
\usepackage{autobreak}
\usepackage{comment}
\usepackage[normalem]{ulem}
\newtheorem{theorem}{Theorem}
\newtheorem{lemma}{Lemma}
\newtheorem{remark}{Remark}
\newtheorem{assumption}{Assumption}

\newtheorem{proposition}{Proposition}

\algnewcommand{\Inputs}[1]{%
  \State \textbf{Inputs:} 
  \parbox[t]{.8\linewidth}{\raggedright #1}
}
\algnewcommand{\Initialize}[1]{%
  \State \textbf{Initialization:}
  \parbox[t]{.8\linewidth}{\raggedright #1}
}

\begin{document} 

\title{\huge Distributed Over-the-air Computing for Fast Distributed Optimization: Beamforming Design and Convergence Analysis}

\author{Zhenyi~Lin, Yi~Gong, and Kaibin~Huang
\thanks{Z. Lin and K. Huang are  with the Dept. of Electrical and Electronic Engr. at The University of Hong Kong, Hong Kong. Z. Lin is also with the Dept. of EEE at Southern University of Science and Technology, China. Y. Gong is with the same department. Contact: K. Huang (Email: huangkb@eee.hku.hk), Y. Gong (Email: gongy@sustech.edu.cn).}

}

\maketitle

\IEEEpeerreviewmaketitle

\begin{abstract}
Distributed optimization concerns the optimization of a common function in a distributed network, which finds a wide range of  applications ranging from machine learning to vehicle platooning. Its key operation is to aggregate all \emph{local state information} (LSI) at devices to update their states. The required extensive message exchange and many iterations cause a communication bottleneck when the LSI is high dimensional or at high mobility. To overcome the bottleneck, we propose in this work the framework of distributed \emph{over-the-air computing} (AirComp) to realize a one-step aggregation for distributed optimization by exploiting simultaneous multicast beamforming of all devices and the property of analog waveform superposition of a multi-access channel. Equivalently, the technique superimposes multiple instances of conventional AirComp processes, giving rise to the challenge of jointly designing multicast beamforming at devices to rein in errors due to interference and channel distortion. We consider two design criteria. The first one is to minimize the sum AirComp error (i.e., sum mean-squared error (MSE)) with respect to the desired average-functional values. An efficient solution  approach is proposed by transforming the non-convex beamforming problem into an equivalent concave-convex fractional program and solving it by nesting convex programming into a bisection search. The second criterion, called zero-forcing (ZF) multicast beamforming, is to force the received over-the-air aggregated signals at devices to be equal to the desired functional values. In this case, the optimal beamforming admits closed form. Both the MMSE and ZF beamforming exhibit a \emph{centroid} structure resulting from averaging columns of conventional MMSE/ZF precoding. Last, the convergence of a classic distributed optimization algorithm is analyzed. The distributed AirComp is found to accelerate convergence by dramatically reducing communication latency. Another key finding is that the ZF beamforming outperforms the MMSE design as the latter is shown to cause bias in subgradient estimation. 
\end{abstract}

\section{Introduction}
Distributed optimization concerns the optimization of a common function in a distributed network comprising a cluster of edge devices connected by \emph{device-to-device} (D2D) links \cite{rabbat2005quantized}. The broad field covers two popular areas: \emph{distributed machine learning} aiming at leveraging local data and computer resources at edge devices to train a global artificial-intelligence (AI) model \cite{zhou2019edge}, and \emph{distributed consensus} finding a wide range of applications in, for example, swarms of drones or robots and vehicle platooning \cite{letaief2021edge}. A fundamental operation in distributed optimization, which is usually implemented using an iterative gradient-descent algorithm, is to aggregate all \emph{local state information} (LSI) and use the result to update the states of all devices in each round. The required extensive message exchange and many rounds cause a communication bottleneck when the LSI is high dimensional (e.g., distributed learning) or at high mobility (e.g., drones). To overcome the bottleneck, we propose in this work the framework of distributed \emph{over-the-air computing} (AirComp) to realize a \emph{one-step} aggregation for distributed optimization by exploiting simultaneous multicast beamforming by all devices and the property of analog waveform superposition of a multi-access channel. To develop the framework, beamforming designs to rein in channel distortion are presented and the convergence of distributed optimization adopting distributed AirComp is studied.

Recent years have witnessed fast-growing interests in distributed machine learning as driven by the distillation of enormous mobile data into artificial intelligence (AI) to power next-generation applications ranging from industrial automation to smart cities and extended reality. Relevant research mainly focuses on the efficient deployment of federated learning, arguably the most popular distributed learning framework, at the network edge, giving rise to the fast growing area of \emph{federated edge learning} (FEEL) \cite{lim2020federated}. FEEL gains popularity for its capabilities of leveraging local data while helping to preserve their privacy and distributed computation resources. Typically implemented in networks with a star topology, the FEEL algorithm iterates between 1) aggregation at an edge server over local models, which are updated by edge devices using local data and uploaded over wireless channels, to update a global model, and 2) broadcast of the model to all devices for updating, until the global model converges. The main challenge confronting efficient FEEL is the communication bottleneck caused by uploading high-dimensional local modes (or stochastic-gradients) from potentially many devices over a multi-access channel. The required large number of rounds/iterations (e.g., tens to hundreds of rounds)  exacerbates the issue. Attempts to overcome the bottleneck have led to the proposal of different techniques   including radio resource management \cite{chen2020joint}, \cite{zeng2020energy}, model quantization \cite{zhu2020one}, \cite{du2020high}, and device scheduling \cite{yang2019scheduling}, \cite{ren2020scheduling}. One particular class of techniques of our interest, known as \emph{over-the-air FEEL}, features the application of AirComp to realize over-the-air model aggregation in FEEL \cite{chen2021distributed, zhu2021over,amiri2020machine,yang2020federated}. The principle underpinning AirComp (as well as over-the-air FEEL) is to exploit the waveform superposition property such that the signal received at the server approximates a desired aggregation function of linear analog modulated data (e.g., local models/gradients) simultaneously transmitted by devices \cite{zhu2021over, zhu_mimo_2019}. 

Most recently, researchers have studied distributed FEEL targeting a cluster of devices without coordination by a sever and connected by D2D links \cite{xing2021federated,shi2021over,ozfatura2020decentralized}. The original techniques for server-assisted FEEL can be adopted by arranging the devices to take turn or use  orthogonal channels to play the role of edge server \cite{shi2021over,savazzi2020joint}. As the resultant sequential aggregation is time consuming, attempts on realizing parallel operations have been made by selecting multiple weakly  coupled clusters of devices to perform simultaneous intra-cluster aggregation \cite{xing2021federated,ozfatura2020decentralized}. Such approaches are ineffective for a single cluster of devices with tightly coupled links such as a drone swarm or a vehicle platoon, and still require multiple time slots to complete aggregation over all devices. Fundamentally, the drawback of the existing approaches is rooted in attempting to orthogonalize multiple aggregation processes in distributed optimization. On the contrary, we advocate fusing them into a single multi-aggregation process to enable a \textit{one-step} distributed aggregation over all devices. Thereby, the resultant design, termed \textit{distributed AirComp}, supports fast distributed FEEL and distributed optimization at large. 

Consider the aggregation operation of distributed optimization among $K$ devices. The goal of designing distributed AirComp is to realize \textit{one-step} updating of the states of all devices via over-the-air aggregation over their LSI. To this end, all devices simultaneously multicast LSI to their peers and at the same time receive aggregated LSI via full-duplex communication \cite{cirik2014achievable}. We propose provisioning devices with transmit antenna arrays to enable multicast beamforming for reining in the sum AirComp error. On the other hand, the use of receive arrays can support aggregation beamforming as studied in \cite{zhu_mimo_2019}. To simplify our design, we assume single receive antenna at devices. The extension of the current design to include aggregation beamforming is straightforward but make the design tedious without new insight. Usually targeting a single-cell system,  traditional multicast beamforming at the base station aims to efficiently deliver information to multiple receivers under their quality-of-service requirements. Mathematically,  the beamforming design can be formulated as an optimization problem with the objective of minimizing the required transmission resources at the base station (e.g., power and array size) under the constraints of receive signal-to-noise ratios (SNRs) \cite{sidiropoulos2006transmit,mehanna2013joint}. Another popular “max-min” formulation targets maximizing the lowest rate or receive SNR among the receivers under a transmission power constraint \cite{sun2004transmit,lee2012new}. Consider the context of distributed AirComp. Define the AirComp error as the deviation of an aggregated signal from the desired average-functional value due to channel fading and noise. Then the multicast beamforming from the perspective of an individual device aims to minimize the sum AirComp error over other devices under a transmission power constraint. Solving such a problem is no more difficult than the mentioned traditional ones. However, the new challenge arises from $K$ simultaneous over-the-air aggregation processes coupling  multicast beamforming at $K$ devices. This necessitates the  joint beamforming  design that should account for all D2D-channel states in the system. To be precise, the multicast radiation patterns of $K$ devices should be jointly optimized under individual power constraints such that their over-the-air superposition leads to the minimization of sum AirComp error. AirComp requires the aggregated signal arriving at a device to be scaled by a factor, called \emph{alignment level}, to balance approaching a desired functional value and noise amplification \cite{cao2020optimized}. The said problem is further complicated by the need of optimization over the alignment level. In contrast, transmit beamforming for single-aggregation AirComp is simple as its main purpose is to overcome the fading of an associated single-user channel, and hence is irrelevant to the current problem \cite{zhu_mimo_2019,wen2019reduced}. 

In this work, besides proposing the principle of distributed AirComp as described earlier, we design distributed multicast beamforming to materialize the principle and further apply the design to distributed optimization. The key contributions and findings are summarized as follows. 

\begin{itemize}
    \item \textit{MMSE beamforming for distributed AirComp}: Consider the design criterion of minimizing the sum AirComp error defined as the sum \textit{mean-squared error} (MSE) over all devices with respect to the desired average-functional values. We propose an efficient approach for optimally solving the mentioned problem of distributed multicast beamforming. Without compromising the optimality, the optimal alignment factor conditioned on beamforming is derived in a closed form and substituted into the original problem. Even though the resultant problem is still non-convex, it can be transformed into an equivalent \text{concave-convex fractional program}. The transformation admits an efficient solution method that nests solving a convex problem into a bisection search. The results reveal that the optimal multicast beamforming at a device is steered along the \emph{centroid} of the column vectors of a traditional MMSE precoder for the associated D2D multiuser channel to the peers, thereby balancing  their AirComp errors; at least one device transmits with full power while some devices transmit with partial power.
    
    \item \emph{Zero-forcing (ZF) beamforming for distributed AirComp}: Consider the design criterion of forcing all receive signal power to approach a uniform level for the purpose of aggregation without attempting to avoid potential noise amplification as in the MMSE case. To minimize the resultant sum AirComp error, the ZF multicast beamformers conditioned on the alignment factor can be reduced into a single-device design with the solution derived in a closed form. It is observed to have a similar centroid  form as the MMSE counterpart but is based on a traditional ZF precoder. Given the beamformers, the alignment factor is obtained as the minimum of a derived expression over devices, which represents an attempt to cope with the weakest set of D2D links limiting the performance of distributed AirComp. 
    
    \item \emph{Convergence of distributed optimization}: The preceding  distributed AirComp framework is applied to implement distributed optimization over D2D links, which is based on the widely-used distributed dual averaging algorithm. The convergence analysis shows that AirComp errors induce bias terms in  the gap between the minimized loss function and its ground truth. The key finding is that the MMSE and ZF designs for distributed multicast beamforming lead to \emph{biased} and \emph{unbiased} estimations of ground-truth stochastic subgradients at devices, respectively. Consequently, at a low-to-medium  receive SNR, the former results in a converged test accuracy substantially lower than that in the ideal (noise-free) case while the loss of the ZF counterpart is negligible. This is opposite to the fact that ZF is sub-optimal in terms of minimizing the sum MMSE AirComp error. For both designs, the said loss is shown to diminish as the transmit SNR grows by being inversely proportional to its square root. 
\end{itemize}

The remainder of this paper is organized as follows. An overview of distributed AirComp is given in  Section II. The MMSE and ZF designs of distributed multicast  beamforming are presened  in Sections III and IV, respectively. The application of distributed AirComp to distributed optimization is studied  in Section V. Experimental results are presented in Section VI, followed by concluding remarks in Section VII.

\section{Overview of Distributed AirComp} \label{sec:overview}
\subsection{System Model} 
As illustrated in Fig. \ref{Fig: dis_sys},  the considered distributed system comprises $K$ edge devices without coordination by a server. Each device is equipped with $N_t$  transmit antennas and one single receive  antenna to support  multicast  beamforming and full-duplex communication. The number of transmit antennas is assumed sufficiently large, i.e., $(N_t\geqslant K-1)$, to provide enough degrees-of-freedom (DoF) for each device to support simultaneous aggregations at $(K-1)$ peers. For full-duplex communication, it is assumed that a  transmitter is perfectly decoupled from a receiver at the same device via passive and/or active self-inference cancellation \cite{zhang2015full} as illustrated in Fig.~\ref{Fig: dis_sys}.   

\begin{figure}[t!]
    \centering
    \includegraphics[width=.7\linewidth]{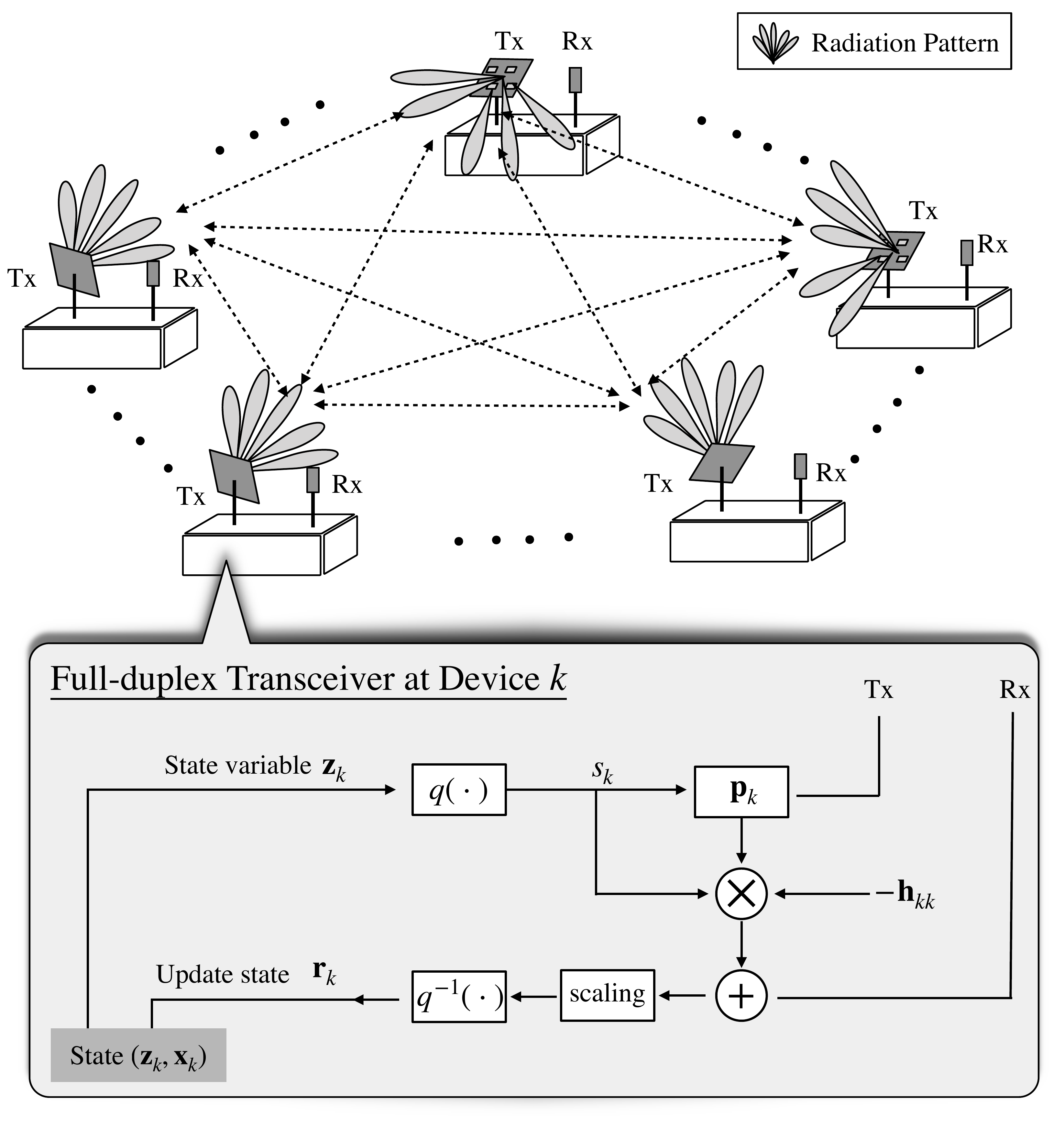}
    \caption{Distributed AirComp system featuring distributed multicast beamforming and full-duplex communication.}
    \label{Fig: dis_sys}
\end{figure}

As the process of  distributed optimization spans multiple rounds, time is divided into rounds with a fixed duration denoted as $T_\text{cmm}$; each round is further divided into $D$ symbol slots. Channels are assumed to be static in each round and varying over rounds. For simplicity, each link is assumed frequency non-selective with the bandwidth represented by $
B$, giving the symbol duration $T_s = \frac{1}{B}$. Global  \textit{channel state information} (CSI) is available at each device via estimation exploiting channel reciprocity or efficient feedback \cite{zhu_mimo_2019}. The required training overhead is assumed negligible compared with high-dimensional LSI exchange. In each round, each device, say device $k$,  transmits a symbol vector, denoted as $\mathbf{s}_k$, that comprises $D$ symbols and spans the round duration of $T_\text{cmm} = DT_s$. 

\subsection{Distributed AirComp Design}

An iterative algorithm for distributed optimization (see Appendix~\ref{App_prelimarary}) comprises multiple communication rounds. In each round, all devices broadcast their LSI simultaneously using linear analog modulation and multicast beamforming, and at the same time receive the over-the-air aggregated signals from other devices. All devices are assumed to be synchronized via a common clock. Consider an arbitrary round, for each device, say device $k$, let $\mathbf{p}_k \in \mathbb{C}^{N_t\times 1}$ denote the multicast beamformer under the power constraint, $\|\mathbf{p}_k\|^2 \leqslant P_0$, with $P_0$ representing the maximum power. Then the over-the-air aggregated signal vector received by the device, denoted as $\mathbf{y}_k$, is given as 
\begin{equation}
    \mathbf{y}_k = \sum\limits_{\ell=1, \ell\neq k}^K \mathbf{h}^{H}_{\ell k}\mathbf{p}_\ell \mathbf{s}_\ell+\tilde{\mathbf{w}}_k,
\end{equation} 
where $\mathbf{s}_\ell$ denotes the symbol row vector transmitted by device $\ell$, $\mathbf{h}_{\ell k} \in \mathbb{C}^{N_t\times 1}$ is the channel vector for the link from device $\ell$ to $k$,  and $\tilde{\mathbf{w}}_k\sim\mathcal{CN}(\mathbf{0},\sigma^2\mathbf{I})$ is channel noise. Upon receiving the aggregated signal, it is scaled as shown below to yield the  desired  aggregation  form: 
\begin{equation}
    \mathbf{r}_k = \frac{\mathbf{y}_k}{(K-1)\sqrt{\eta}},
\end{equation}
where $\eta$ is the alignment factor that helps to reduce the AirComp error. 

Consider round $n$ and device $k$. The transmitted symbol vector is a linear analog modulated local state variable ($D\times 1$ vector), $\mathbf{z}_k(n)$ (see Appendix~\ref{App_prelimarary}). Normalization of the variable is necessary for meeting the transmission power constraint and avoiding a nonzero DC level, which unnecessarily increases the power. For distributed optimization with a general objective function, typically the statistics of the elements of $\mathbf{z}_k(n)$ are different and also vary over rounds (see e.g., \cite{zhang2021gradient}). Define the round-dependent mean and variance of local state variables, which are assumed uniform for all devices, as $M(n) = \mathsf{E}\left[\frac{1}{D}\sum_d z^{(d)}_k(n)\right]$ and $V^2(n) = \mathsf{Var}\left[\frac{1}{D}\sum_d z^{(d)}_k(n)\right]$. These parameters can be estimated offline and stored at devices \cite{zhang2021gradient}.  Let $q(\cdot)$ denote the normalization function such as $\mathbf{s}_k(n) = q(\mathbf{z}_k(n)) = \frac{\mathbf{z}_k(n)-M(n)}{V(n)}$ for all $k$.  Consequently,  $\mathsf{Var}\left[\frac{1}{D}\sum_d s^{(d)}_k(n)\right] = 1$ and $\mathsf{E}[\frac{1}{D}\sum_d s^{(d)}_k(n)]=0$ with $s^{(d)}_k(n)$ being the $d$th element of $\mathbf{s}_k(n)$. Note that $\mathsf{E}[\mathbf{s}_k(n)] \neq \mathbf{0}$ due to the non-uniform distributions of the elements in $\mathbf{s}_k(n)$. Then the  received signal is de-normalized as 
\begin{equation}\label{eq:vec_rec}
    \mathbf{r}_k = q^{-1}\left(\frac{\mathbf{y}_k}{(K-1)\sqrt{\eta}}\right) = \frac{V}{(K-1)\sqrt{\eta}}\sum_{\ell =1, \ell \neq k}^K\mathbf{h}_{ \ell k  }^{H} \mathbf{p}_{\ell} \mathbf{s}_\ell+M+\mathbf{w}_k,
\end{equation}
where the channel noise $\mathbf{w}_k\sim \mathcal{CN}\left(\mathbf{0},\frac{V^2\sigma^2}{(K-1)^2\eta}\mathbf{I}\right)$.

\subsection{Distributed Optimization and Performance Metric}

In distributed optimization, there exist local objective functions associated with individual  devices. The function for device $k$ is denoted as  $f_{k}: \mathbb{R}^{D} \rightarrow \mathbb{R}$.
 Following the common assumptions, each function is assumed to be convex and hence sub-differentiable, but not necessarily smooth. The global objective, denoted as $f$, is related to local objectives by $f(\mathbf{x}) = \frac{1}{K} \sum_{k=1}^K f_{k}(\mathbf{x})$ for a common state of devices. Then the goal of distributed optimization is to find the optimal common state $\mathbf{x}^\star$, which represents a consensus among devices, that minimizes the objective. Mathematically, 
\begin{equation}\label{eq:global_obj}
    \mathbf{x}^\star = \arg\min _{\mathbf{x} \in \mathcal{X}} f(\mathbf{x}) = \arg\min _{\mathbf{x} \in \mathcal{X}} \frac{1}{K} \sum_{k=1}^K f_{k}(\mathbf{x}),
\end{equation}
where the state space  $\mathcal{X}$ is a closed, $D$-dimensional convex set. 

A  classic algorithm, distributed dual averaging, for solving \eqref{eq:global_obj} is discussed in  Appendix~\ref{App_prelimarary}. In this algorithm, the local state of each device, say device $\ell$,  contains a primal variable and a dual variable, which are denoted as $\mathbf{x}_\ell(n)\in \mathbb{R}^D$ and $\mathbf{z}_\ell(n)\in \mathbb{R}^D$ for round $n$, respectively. Each round, say round $n$,  of the algorithm comprises a three-step procedure  simultaneously executed at each device, say device $\ell$: (1) calculate a stochastic subgradient $\tilde{\bm{g}}_\ell(n)$ of its local loss function at $\mathbf{x}_\ell(n)$; (2)  aggregate the dual variables of  all other devices and apply the result together with $\tilde{\bm{g}}_\ell(n)$ to update the local dual variable as $\mathbf{z}_\ell(n+1)$; (3) project $\mathbf{z}_\ell(n+1)$ onto the feasible set $\mathcal{X}$ to obtain the updated  primal variable  $\mathbf{x}_\ell(n+1)$ \cite{duchi_dual_2012}. Only Step (2) requires information exchange among devices. We focus on its one-step implementation using distributed AirComp in the preceding subsection. To this end, \eqref{eq:step:1} for Step (2) at device $\ell$ is modified as \begin{equation}\label{eq:dual_updates}
    \mathbf{z}_\ell(n+1) = (1-\beta)\mathbf{z}_\ell(n) + \beta \mathbf{r}_\ell(n) + \tilde{\bm{g}}_\ell(n)
\end{equation}
where the received signal $\mathbf{r}_\ell(n)$ is given in \eqref{eq:vec_rec} and $\beta$ is a constant weight (see Appendix~\ref{App_prelimarary}). Note that the ideal received signal, namely the ground truth,  is  $\frac{1}{K-1}\sum\limits_{k\neq \ell}^K \mathbf{z}_k(n)$. Then the sum AirComp error is suitably defined as the sum  \textit{mean-square-error} (MSE) between the received signal and its ground-truth: 
\begin{equation}
    \begin{aligned}[b]\label{ori_pro}
        \mathrm{MSE} &=\mathsf{E}\left[\sum_{\ell=1}^{K}\sum_{d=1}^{D}\left(r_\ell^{(d)}-\frac{1}{K-1} \sum_{k\neq \ell}^{K} z_{k}^{(d)}\right)^{H}\left(r_\ell^{(d)}-\frac{1}{K-1} \sum_{k\neq \ell}^{K} z_{k}^{(d)}\right)\right] \\
        &=\frac{V^2D}{(K\!-\!1)^{2}}\mathsf{E}\left[\sum_{\ell=1}^{K}\left(\sum_{k\neq \ell}^{K}\frac{\mathbf{h}_{k \ell }^{H} \mathbf{p}_{k}}{\sqrt{\eta}}s_k^{(d)}\!-\!\sum_{k\neq \ell}^{K} s_{k}^{(d)}\!+\!{w}_k^{(d)}\right)^{H}\left(\sum_{k\neq \ell}^{K}\frac{\mathbf{h}_{k \ell }^{H} \mathbf{p}_{k}}{\sqrt{\eta}}s_k^{(d)}\!-\!\sum_{k\neq \ell}^{K} s_{k}^{(d)}\!+\!{w}_k^{(d)}\right)\right] \\
 &=\frac{V^2D}{(K-1)^{2}} \sum_{\ell=1}^{K}\left[\sum_{k\neq \ell}^{K}\left\|\frac{\mathbf{h}_{k \ell }^{H} \mathbf{p}_{k}}{\sqrt{\eta}} -1\right\|^{2}+\frac{\sigma^{2}}{ \eta}\right] \\
        &=\frac{V^2D}{(K-1)^{2}} \sum_{k=1}^{K} \sum_{\ell\neq k}^{K}\left\|\frac{\mathbf{h}_{k \ell }^{H} \mathbf{p}_{k}}{\sqrt{\eta}} -1\right\|^{2}+\frac{KDV^2\sigma^{2}}{(K-1)^{2}} \frac{1}{\eta},
    \end{aligned}
\end{equation}
where the expectation is taken over the distribution of the transmitted symbol and the received noise $w_k^{(d)}$ and the round  index $(n)$ is omitted here for brevity.

\section{MMSE Design of Distributed Multicast Beamforming}\label{sec:MMSE}
In this section, distributed multicast  beamforming for distributed AirComp as well as the alignment factor are  jointly optimized under the criterion of minimum sum AirComp error. The result is  called  minimum MSE (MMSE) multicast beamforming. An efficient solution approach is developed  based on fractional programming. 

\subsection{Optimal Distributed Multicast  Beamforming}
Consider an arbitrary round with its index omitted for ease of notation. With the sum AirComp error in \eqref{ori_pro}, the mentioned optimization problem is formulated as follows 
\begin{equation}
    \begin{aligned}
       &\min_{\{\mathbf{p}_k\}, \eta}&& \text{MSE}(\mathbf{p}_1, \cdots, \mathbf{p}_k, \eta) \\
        &\mathrm{s.t.} && \|\mathbf{p}_k\|^2\leqslant P_0,~~\forall k,\\
        &&& \eta \geq 0.
    \end{aligned} 
    \tag{P1}\label{P1}
\end{equation}
Due to the coupling of $\{\mathbf{p}_k\}$ and $\eta$, Problem \eqref{P1} is non-convex and hence directly solving it  is difficult. To overcome the difficulty, an alternative approach is presented as follows. 

First, given $\{\mathbf{p}_k\}$, Problem \eqref{P1} is reduced to the following conditional problem: 
\begin{equation}\label{eq:given_pi}
    \begin{aligned}
     \min_{\eta\geq 0}\quad \frac{V^2D}{(K-1)^{2}} \sum_{k=1}^{K} \sum_{ \ell\neq k}^{K}\left\|\frac{\mathbf{h}_{k \ell }^{H} \mathbf{p}_{k}}{\sqrt{\eta}} -1\right\|^{2}+\frac{KDV^2\sigma^{2}}{(K-1)^{2}} \frac{1}{\eta}.
    \end{aligned}
\end{equation}
The objective in \eqref{eq:given_pi} is found to be convex and then the optimal $\eta^{\star}_\text{mmse}$ is obtained as
\begin{equation}\label{eq:given_pi_eta}
    \eta^{\star}_\text{mmse} = \left(
    \frac{2K\sigma^2+2\sum_{\ell=1}^K\sum_{k\neq \ell}^K \mathbf{p}_k^H\mathbf{h}_{k \ell}\mathbf{h}_{k \ell}^H\mathbf{p}_k}{\sum_{\ell = 1}^K\sum_{k\neq \ell}^K( \mathbf{h}_{k \ell}^H\mathbf{p}_k+\mathbf{p}_k^H\mathbf{h}_{k \ell})}
    \right)^2.
\end{equation}
Next, we proceed to optimize $\{\mathbf{p}_k\}$. Substituting $\eta^{\star}_\text{mmse}$ in \eqref{eq:given_pi_eta} into Problem \eqref{P1} gives 
\begin{equation}
\begin{aligned}
\min _{\left\{\mathbf{p}_{k}\right\}} 
& \quad \frac{K}{K-1}- \frac{1}{(K-1)^2}\cdot\frac{\left(\sum_{\ell = 1}^K\sum_{k\neq \ell}^K( \mathbf{h}_{k \ell}^H\mathbf{p}_k+\mathbf{p}_k^H\mathbf{h}_{k \ell})\right)^{2}}{4(K\sigma^{2}+\sum_{\ell=1}^{K}\sum_{k\neq \ell}^{K}\mathbf{h}_{k \ell }^{H} \mathbf{p}_{k}\mathbf{p}_{k}^{H}\mathbf{h}_{k \ell })} \\
\text { s.t. }& \quad\left\|\mathbf{p}_{k}\right\|^{2} \leqslant P_{0}, \quad \forall k.
\end{aligned}
\tag{P1.1} \label{P1.1}
\end{equation}
Problem \eqref{P1.1} remains non-convex. Nevertheless, the result in the following lemma can transform this problem into a tractable equivalent form.

\begin{lemma}\label{lemma:ccfp}\emph{
Given $\sum_{\ell = 1}^K\sum_{k\neq \ell}^K( \mathbf{h}_{k \ell}^H\mathbf{p}_k+\mathbf{p}_k^H\mathbf{h}_{k \ell})\geqslant 0$, Problem \eqref{P1.1} is equivalent to the following \textit{concave-convex fractional program}:
\begin{equation}
\begin{aligned}
\max_{\left\{\mathbf{p}_{k}\right\}} 
& \quad \frac{\left(\sum_{\ell = 1}^K\sum_{k\neq \ell}^K( \mathbf{h}_{k \ell}^H\mathbf{p}_k+\mathbf{p}_k^H\mathbf{h}_{k \ell})\right)}{2\sqrt{(K\sigma^{2}+\sum_{\ell=1}^{K}\sum_{k\neq \ell}^{K}\mathbf{h}_{k \ell }^{H} \mathbf{p}_{k}\mathbf{p}_{k}^{H}\mathbf{h}_{k \ell })}} \\
\text { s.t. }
& \quad\left\|\mathbf{p}_{k}\right\|^{2} \leqslant P_{0} \quad \forall k, \\ 
& \quad \sum_{\ell = 1}^K\sum_{k\neq \ell}^K( \mathbf{h}_{k \ell}^H\mathbf{p}_k+\mathbf{p}_k^H\mathbf{h}_{k \ell})\geqslant 0.
\end{aligned}
\tag{P1.2}\label{P1.2}
\end{equation}
}
\end{lemma}
\begin{proof}
    See Appendix \ref{App_ccfp}.
\end{proof}
It can be straightforwardly proved by a simple variable substitution. When $\sum_{\ell = 1}^K\sum_{k\neq \ell}^K( \mathbf{h}_{k \ell}^H\mathbf{p}_k+\mathbf{p}_k^H\mathbf{h}_{k \ell})< 0$, the optimal solution of \eqref{P1.1} can be written as $\tilde{\mathbf{p}}_k^{\star} = -\mathbf{p}_k^{\star}$, where $\mathbf{p}_k^{\star}$ is the optimal solution of \eqref{P1.2}. Since $\tilde{\mathbf{p}}_k^{\star}$ and $ \mathbf{p}_k^{\star}$ correspond to the same optimal value of the objective of \eqref{P1.1}, considering the case in Lemma \ref{lemma:ccfp} where $\sum_{\ell = 1}^K\sum_{k\neq \ell}^K( \mathbf{h}_{k \ell}^H\mathbf{p}_k+\mathbf{p}_k^H\mathbf{h}_{k \ell}) \geqslant 0$ is sufficient for the purpose of solving  Problem \eqref{P1.1}. This allows a property of factional program to be applied. To this end, define the \textit{upper contour set} of a function $f:\mathbb{R}^n\rightarrow \mathbb{R}$ at $c = f(x_0)$ for some $x_0\in\mathbb{R}^n$ as the set
$\left\{ x\in\mathsf{X}:~f(x)\geqslant c\right\}$.  One useful property of the concave-convex fractional program in Lemma \ref{lemma:ccfp} is its strict quasi-concavity\cite{schaible1976fractional}. Specifically,  the upper contour sets of the objective of Problem \eqref{P1.2}, are \emph{convex}. By introducing an auxiliary variable $\alpha (\alpha>0)$, which corresponds to the aligned fraction, Problem \eqref{P1.2} can be written as the convex problem of maximization over $\alpha$ as follows: 
\begin{equation}
\begin{aligned}
\max _{\alpha,\left\{\mathbf{p}_{k}\right\},} 
& \quad \alpha \\
\text { s.t. }
&
\quad\left\|\mathbf{p}_{k}\right\|^{2} \leqslant P_{0}, \quad \forall k, \\
&
\quad \alpha\leqslant\frac{\left(\sum_{\ell = 1}^K\sum_{k\neq \ell}^K( \mathbf{h}_{k \ell}^H\mathbf{p}_k+\mathbf{p}_k^H\mathbf{h}_{k \ell})\right)}{2\sqrt{(K\sigma^{2}+\sum_{\ell=1}^{K}\sum_{k\neq \ell}^{K}\mathbf{h}_{k \ell }^{H} \mathbf{p}_{k}\mathbf{p}_{k}^{H}\mathbf{h}_{k \ell })}}.
\end{aligned}
\tag{P1.3}\label{P1.3}
\end{equation}

To solve the above problem requires consideration of a related problem of transmission power minimization described  as follows. Define the maximum power as $p_{\text{max}} = \max_k \|\mathbf{p}_k\|_2^2$. Then given $\alpha$, the mentioned problem is given as 
\begin{equation}
\begin{aligned}
\min _{\{\mathbf{p}_{k}\}, p_{\text{max}}}
& \quad p_{\text{max}} \\
\text { s.t. }&
\quad 2\alpha\sqrt{(K\sigma^{2}+\sum_{\ell=1}^{K}\sum_{k\neq \ell}^{K}\mathbf{h}_{k \ell }^{H} \mathbf{p}_{k}\mathbf{p}_{k}^{H}\mathbf{h}_{k \ell })}\leqslant \left(\sum_{\ell = 1}^K\sum_{k\neq \ell}^K( \mathbf{h}_{k \ell}^H\mathbf{p}_k+\mathbf{p}_k^H\mathbf{h}_{k \ell})\right), \\
& 
\quad
\|\mathbf{p}_k\|_2^2 \leqslant p_{\text{max}}, \quad \forall k.
\end{aligned}
\tag{P1.4}\label{P1.4}
\end{equation}
The solution of Problem~\eqref{P1.4} is denoted as a function of $\alpha$, $p^{\star}(\alpha)$. Note that given the desired aligned fraction $\alpha$, Problem \eqref{P1.4} is feasible if and only if the solution for Problem \eqref{P1.4} satisfies $p^{\star}(\alpha)\leqslant P_0$. Two useful lemmas for relating Problems \eqref{P1.3} and \eqref{P1.4} are given as follows, which are proved in Appendices \ref{App_mon_proof} and \ref{App_con_proof}, respectively.

\begin{lemma}\label{lemma:mon_proof} \emph{The minimum transmission  power $p^{\star}(\alpha)$ over devices is a monotonously increasing function of $\alpha$.
}
\end{lemma}

It follows from the result in Lemma \ref{lemma:mon_proof} that the solution for Problem~\eqref{P1.3} is the maximum
alignment level, denoted as $\alpha^{\star}$,  such that  the corresponding minimum transmission power $p^{\star}(\alpha^{\star})$ is no larger than $P_0$. This suggests a solution method of Problem \eqref{P1.3} by a search for $\alpha^{\star}$ over the range of  $0<   p^{\star}(\alpha) \leqslant P_0$. This requires solving Problem \eqref{P1.4} so as to compute the function $p^{\star}(\alpha)$. To this end, the following result is useful.
\begin{lemma}\label{lemma:con_proof} \emph{Given $\alpha$, Problem \eqref{P1.4} is convex.}
\end{lemma}
The convexity of Problem \eqref{P1.4} allows it to be solved efficiently using rich existing techniques for convex programming, for example,  the primal-dual method. 

In summary, Problem~\eqref{P1.3} for optimal transmit beamforming can be solved efficient by nesting a one-dimensional search over $\alpha$ and the solution of the convex  Problem \eqref{P1.4}. The detailed  algorithm is presented  in Algorithm~\ref{Alg:bisec}.

\begin{algorithm}[t!]
    \setstretch{1.5}
    \caption{Optimal Algorithm for MMSE Distributed Multicast  Beamforming}
    \label{Alg:bisec}
    \begin{algorithmic}[1]
      \Inputs{$K, \sigma, \{\mathbf{h}_{k \ell}\}$.}
      \Initialize{}
      \Statex Select $\alpha_\text{u}$ so that $\alpha = \alpha_\text{u}$ makes $p^{\star}(\alpha_\text{u})$ defined in \eqref{P1.4} larger than $P_0$.
      \Statex Select $\alpha_\text{l}$ so that $\alpha = \alpha_\text{l}$ makes $p^{\star}(\alpha_\text{l})<P_0$. 
      \While{$|\alpha_\text{u}-\alpha_\text{l}|\geqslant \epsilon$}
      \State Let $\alpha_\text{m} = (\alpha_\text{u}+\alpha_\text{l})/2$ and substitute $\alpha = \alpha_\text{m}$ into \eqref{P1.4}.
      \State Solve \eqref{P1.4} by CVX toolbox to obtain $p^{\star}(\alpha_\text{m})$ and $\{\mathbf{p}_k^{\star}\}$.
      \If{$p^{\star}(\alpha_\text{m})>P_0$} 
      \State $\alpha_\text{u} = \alpha_\text{m}$.
      \Else 
      \State $\alpha_\text{l} = \alpha_\text{m}$.
      \EndIf
      \EndWhile
      \State \Return $\{\mathbf{p}_k^{\star}\}$.
    \end{algorithmic}
\end{algorithm}

\subsection{Centroid Beamforming}
To shed light on the structure of the optimal MMSE design of distributed multicast beamformer, we consider a special case when the number of antennas $N_t=K-1$. Let $\mathbf{H}_k \in \mathbb{C}^{N_t\times(K-1) }$ be the matrix with its $\ell$-th column being  $\mathbf{h}_{k \ell}^H$ with $\ell\neq k$. In the current case, $\mathbf{H}_k$ is an invertible full-rank matrix. The convexity of Problem \eqref{P1.4} allows its KKT conditions to  be a sufficient and necessary condition for its optimal solution. Based on the KKT conditions,  Lemmas \ref{lemma:one_full} and \ref{lemma:direc} are obtained as follows with  proofs  in Appendix \ref{App_two_lemma}.
\begin{lemma}\label{lemma:one_full} \emph{At least one device performs full-power transmission, i.e. $\exists k,~ \|\mathbf{p}_k\|^2 = P_0$}.
\end{lemma}

\begin{lemma}[Centroid Beamforming]  \label{lemma:direc} \emph{Define $[\mathbf{H}]_\ell$ as the $\ell$-th column of matrix $\mathbf{H}$. When $N_t=K-1$,
the beamforming directions of devices with  partial-power transmission are given  as
\begin{equation}\label{eq:dir_par}
    \frac{\mathbf{p}_k^{\star}}{\|\mathbf{p}_k^{\star}\|} = \frac{\sum\limits_{\ell\neq k}^K\left[\left(\mathbf{H}_k^H\right)^{-1}\right]_\ell}{\left\|\sum\limits_{\ell\neq k}^K\left[\left(\mathbf{H}_k^H\right)^{-1}\right]_\ell\right\|},
\end{equation}
and those with  full-power transmission are given as
\begin{equation}\label{eq:dir_full}
    \frac{\mathbf{p}_k^{\star}}{\|\mathbf{p}_k^{\star}\|} = \frac{\sum\limits_{\ell\neq k}^K\left[\left(\mathbf{H}_k\mathbf{H}_k^H+\mu_k\mathbf{I}\right)^{-1}\mathbf{H}_k\right]_\ell}{\left\|\sum\limits_{\ell\neq k}^K\left[\left(\mathbf{H}_k\mathbf{H}_k^H+\mu_k\mathbf{I}\right)^{-1}\mathbf{H}_k\right]_\ell\right\|},
\end{equation}
where $\mu_k$ is the regularization term that diminishes as the SNR increases. 
}
\end{lemma}

 The result in \eqref{eq:dir_par} shows that the direction of the optimal beamforming of a device with partial-power transmission points to the centroid of the column vectors of a ZF or a regularized channel-inversion precoder. Such a design overcomes fading of multiuser channels to  facilitate simultaneous signal alignments at other devices while balancing their AirComp errors.

\section{Zero-forcing Design of  Distributed Multicast Beamforming}\label{sec:ZF}
In this section, we consider the ZF design criterion of forcing the received signals to approach the desired ground-truth values without consideration of channel noise. Under this criterion, a low-complexity design of distributed multicast beamforming is presented as follows. Conditioned on the alignment factor  $\eta$, the ZF criterion allows the joint beamforming design to be decoupled into the following individual beamforming problems formulated for device $k$ as 
\begin{equation}\label{eq:zf_min}
    \mathbf{p}_k^{\star} = \arg\min_{\mathbf{p}_k}\quad \left\|\mathbf{H}_k^H\mathbf{p}_k-\sqrt{\eta}\mathbf{1}_{(K-1)}\right\|^2,~\forall k,
    \tag{P2}
\end{equation}
where the objective function is termed as \emph{misalignment error}.  
Note that in this case, since $N_t\geqslant K-1$, there is sufficient DoF to align the channels, i.e., $\mathbf{H}_k^H\mathbf{p}_k=\sqrt{\eta}\mathbf{1}_{(K-1)},~\forall k$ is always satisfied. Therefore $\left\|\mathbf{H}_k^H\mathbf{p}_k-\sqrt{\eta}\mathbf{1}_{(K-1)}\right\|^2 = 0$ always has non-trivial solutions. Then following lemma gives a solution with the smallest norm.

\begin{proposition}\label{prop:suff_knt}\emph{ When $N_t\geqslant K-1$, given the alignment factor $\eta$, the misalignment error of device $k$ is minimized by following ZF multicast  beamforming: 
\begin{equation}\label{eq:zf_p2}
    \begin{aligned}
        \mathbf{p}_k = \sqrt{\eta} \mathbf{H}_{k}\left(\mathbf{H}_{k}^{H} \mathbf{H}_{k}\right)^{-1} \mathbf{1}_{(K-1)},\qquad ~\forall k, 
    \end{aligned}
\end{equation}
where $\eta$ is chosen to meet the power constraint. 
}
\end{proposition}
\begin{proof}
    See Appendix \ref{App_suff_knt}.
\end{proof}

The beamformer in   \eqref{eq:zf_p2} can be written in a centroid  form similarly as its MMSE counterpart, i.e., $\mathbf{p}_k \!=\! \sqrt{\eta} \sum\limits_{\ell\!\neq\! k}^K \left[ \mathbf{H}_{k}\left(\mathbf{H}_{k}^{H} \mathbf{H}_{k}\right)^{\!-\!1}\right]_\ell$. Specifically,  this multicast beamforming is given by the centroid of the column vectors of a traditional ZF precoder for the associated D2D multiuser channel to the peers. However, it should be noted  that compared with MMSE beamforming, this ZF scheme can potentially  amplify the noise and is therefore vulnerable to deep fading. 

By substituting \eqref{eq:zf_p2} in Proposition \ref{prop:suff_knt}, the problem of alignment-error minimization over the alignment factor $\eta$ is 
\begin{equation}
    \begin{aligned}
    \min _{\left\{\eta\right\}} 
    &\quad \frac{KD V^2\sigma^2}{(K-1)^2\eta},
    \\
    \text{s.t.}&\quad \eta \mathbf{1}_{(K-1)}^H \left(\mathbf{H}_{k}^{H} \mathbf{H}_{k}\right)^{-1}\mathbf{1}_{(K-1)}\leqslant P_0, \quad \forall k.
    \end{aligned}
    \tag{P2.1}\label{P2.1}
\end{equation}
Since the objective of Problem \eqref{P2.1} is a monotonically decreasing function of $\eta$, the optimal $\eta^{\star}_{\text{ZF}}$ is directly given as
\begin{equation}\label{eq:opt_eta}
    \eta^{\star}_{\text{ZF}} = \min_{k} \frac{P_0}{\mathbf{1}_{(K-1)}^H \left(\mathbf{H}_{k}^{H} \mathbf{H}_{k}\right)^{-1}\mathbf{1}_{(K-1)}}.
\end{equation}
By constructing a Rayleigh quotient as following, we have
\begin{equation}
\begin{aligned}
    \eta^{\star}_{\text{ZF}} & = \min_{k} \frac{P_0\mathbf{1}_{(K-1)}^H\mathbf{1}_{(K-1)}}{(K-1)\mathbf{1}_{(K-1)}^H \left(\mathbf{H}_{k}^{H} \mathbf{H}_{k}\right)^{-1}\mathbf{1}_{(K-1)} } 
    \geqslant
    \min_{k} \frac{P_0}{(K-1)\lambda_{\text{max}} \left( \left(\mathbf{H}_{k}^{H} \mathbf{H}_{k}\right)^{-1}\right) } \\
    & \overset{(a)}{=} \min_{k} \frac{P_0}{(K-1)\lambda_{\text{min}}^{-1}\left( \mathbf{H}_{k}^{H} \mathbf{H}_{k}\right) },
\end{aligned}
\end{equation}
where $(a)$ is due to that $\mathbf{H}_{k}^{H} \mathbf{H}_{k}$ is a hermitian matrix and $\lambda_{\text{min}}(.)$, $\lambda_{\text{max}}(.)$ correspond to the minimum and maximum eigenvalue of given matrix separately. Using this upper bound on  $\eta^{\star}_{\text{ZF}}$ leads to an upper bound on the sum AirComp error in the current case: 
\begin{equation} \label{eq:ZF_mse}
        \text{MSE}_{\text{ZF}} \leqslant \frac{K D V^2 \sigma^2}{(K-1)P_0\min_{k}\lambda_{\text{min}}\left( \mathbf{H}_{k}^{H} \mathbf{H}_{k}\right)}. 
\end{equation}
Note that the minimum eigenvalue $\lambda_{\text{min}}\left( \mathbf{H}_{k}^{H} \mathbf{H}_{k}\right)$ is a monotonically increasing function of $N_t$ \cite{wen2019reduced}. This demonstrates the diversity gain of increasing the transmit array size for reducing the sum AirComp error. Next, comparing the sum AirComp errors of the MMSE and ZF designs, it is worth mentioning that $\text{MSE}_\text{mmse}\leqslant \text{MSE}_\text{ZF}$.

\section{Application to Distributed Optimization} \label{sec:DO}
In this section, we consider the application of distributed AirComp designed in the preceding sections to provide an efficient air interface for implementing the distributed-optimization algorithm in Appendix~\ref{App_prelimarary} in a D2D network. The effects of AirComp error on its convergence are charaterized mathematically. For tractable analysis, several assumptions commonly made in the literature (see e.g., \cite{saha2021decentralized}) are also adopted in this work. 

\begin{assumption}[Continuity] \emph{Each local objective function is L-Lipschitz continuous:
    \begin{equation}
        \left|f_{k}(\mathbf{x})-f_{k}(\mathbf{y})\right| \leq L\|\mathbf{x}-\mathbf{y}\| \quad \text{for}~ \mathbf{x}, \mathbf{y} \in \mathcal{X}
    \end{equation}
    where $L$ is a constant. 
}    
\end{assumption}

\begin{assumption}[Local Gradient Estimation] \emph{The stochastic subgradient $\tilde{\bm{g}}_k(n)$ at the $k$-th device is an \emph{unbiased} estimate  of the ground-truth local subgradient, $\bm{g}_{k}(n) \in \partial f_{k}\left(\mathbf{x}_{k}(n)\right)$,  with bounded second moment:
    \begin{equation}
        \mathsf{E}[\tilde{\bm{g}}_{k}(n)]= \bm{g}_{k}(n) \quad \text { and } \quad \mathsf{E}\left[||\tilde{\bm{g}}_{k}(n)||^2\right] \leq \Omega^2,
    \end{equation}
    where $\Omega$ is a given constant.}
\end{assumption}

\subsection{Effects of AirComp Error}
In this subsection, we investigate  the effects of AirComp error on the subgradient estimation and state updating, which are the key operations of distributed optimization. The variables in the subsequent analysis follow those defined in Appendix~\ref{App_prelimarary}. 

First, consider round $n$, the receive signal vector at device $k$  in \eqref{eq:vec_rec} can be rewritten as 
\begin{equation}
    \mathbf{r}_{k}(n) = \frac{1}{K-1}\sum\limits_{\ell = 1,\ell\neq k}^K\mathbf{z}_{\ell}(n)+\mathbf{\Delta}_k(n),
\end{equation}
where the random variable $\mathbf{\Delta}_k(n)$ defined below represents  the distortion caused by wireless propagation: 
\begin{equation}\label{eq:delta_k}
    \mathbf{\Delta}_k(n) = \frac{ V(n)}{K-1}\sum_{\ell=1, \ell\neq k}^K\left(\frac{\mathbf{h}_{k \ell}^H(n)\mathbf{p}_k(n)}{\sqrt{\eta(n)}}-1\right)\mathbf{s}_\ell(n)+ \mathbf{w}_k(n).
\end{equation}
Note that the summation of $\|\mathbf{\Delta}_k(n)\|^2$ over devices gives the sum AirComp error in \eqref{ori_pro}.  Then based on \eqref{eq:dual_updates} and \eqref{eq:step},  the state variable  at device $i$ is updated as 
\begin{equation}
\begin{aligned}\label{eq:air_upd}
    \mathbf{z}_{k}(n+1) =  \sum_{\ell = 1}^K W_{k \ell } \mathbf{z}_{\ell}(n)+\hat{\bm{g}}_{k}(n),
\end{aligned}
\end{equation}
where $\hat{\bm{g}}_{k}(n)$ represents  the channel distorted version  of the stochastic subgradient, $\hat{\bm{g}}_{k}(n)$, namely  $\hat{\bm{g}}_{k}(n) = \tilde{\bm{g}}_{k}(n)+\beta\mathbf{\Delta}_k(n)$.

Next, the effects of AirComp error are reflected by the deviation of the noisy   subgradient, $\hat{\bm{g}}_{k}(n)$,  from its  ground truth. Consider  ZF beamforming designed in Section \ref{sec:ZF}. From its definition, the channel distortion in  \eqref{eq:delta_k}, $\mathbf{\Delta}_k(n)$, reduces to channel noise: $\mathbf{\Delta}_k(n) = \mathbf{w}_k(n)$. It follows from this fact and Assumption 1 that the noisy subgradient $\hat{\bm{g}}_{k}(n)$ yields an \emph{unbiased estimate} of the ground truth: 
\begin{equation}\label{eq:exp_sub_g_zf}
    \mathsf{E}\left[\hat{\bm{g}}_k(n)\right] = \mathsf{E}\left[\tilde{\bm{g}}_{k}(n)+\beta\mathbf{\Delta}_k(n)\right] = \bm{g}_{k}(n). 
\end{equation}
Consider  MMSE beamforming designed in Section \ref{sec:MMSE}. It follows from \eqref{eq:delta_k} and its definition, the corresponding noisy subgradient, however, is a \emph{biased estimate} of the ground truth:  
\begin{equation}\label{eq:exp_sub_g_mmse}
    \begin{aligned}
        \mathsf{E}\left[\hat{\bm{g}}_k(n)\right] &= \mathsf{E}\left[\tilde{\bm{g}}_{k}(n)+\beta\mathbf{\Delta}_k(n)\right] 
        & = \bm{g}_{k}(n) + \frac{\beta V(n)}{K-1}\sum\limits_{\ell=1, \ell\neq k}^K\left(\frac{\mathbf{h}_{k \ell}^H(n)\mathbf{p}_k(n)}{\sqrt{\eta(n)}}-1\right)\mathsf{E}\left[\mathbf{s}_\ell(n)\right], 
    \end{aligned}
\end{equation}
where $\mathsf{E}\left[\mathbf{s}_\ell(n)\right]\neq 0 $ as discussed in Section \ref{sec:overview}. Even though the MMSE beamforming  achieves lower  AirComp error as defined in \eqref{ori_pro} than the ZF counterpart, the former's biased subgradient estimation leads to worse convergence performance as elaborated in the next subsection. 

Last, a result useful for convergence analysis is obtained by bounding the deviation of the state variable (i.e., dual variable), $\mathbf{z}_{k}(n)$, from its average $\overline{\mathbf{z}}(n) = \frac{1}{K}\sum_k \mathbf{z}_{k}(n)$, termed \emph{dual-variable deviation}. Note that it also serves as a measure of the deviation of individual noisy subgradients. Mathematically, the deviation is defined as $\mathsf{E}\left\|\overline{\mathbf{z}}(n)-\mathbf{z}_{k}(n)\right\|_*$ with $\|\cdot\|_* = \mathrm{sup}_{\|u=1\|}\langle.,u\rangle$ denoting the dual norm. To bound the dual-variable deviation,  based on Assumption~2, the second moment of the noisy sub-gradient, $\hat{\bm{g}}_{k}(n)$, can be bounded as:
\begin{equation}\label{eq:var_sub_g}
    \mathsf{E}\left[\|\hat{\bm{g}}_k(n)\|^2\right] = \mathsf{E}\left[\|\tilde{\bm{g}}_{k}(n)\|^2+\beta^2\|\mathbf{\Delta}_k\|^2\right] \leqslant \Omega^2+\frac{\beta^2\text{MSE}(n)}{K}, 
\end{equation}
where the AirComp error $\text{MSE}(n)$ defined in \eqref{ori_pro} applies to both the cases of ZF and MMSE beamforming. Using the bound, the desired result is obtained as follows. 
\begin{lemma}\label{lemma:est_err} \emph{For distributed optimization using distributed AirComp (with either MMSE or ZF beamforming), the expected dual-variable deviation can be  bounded as:
    \begin{equation}\label{eq:cons_err}
        \mathsf{E}\left\|\overline{\mathbf{z}}(n)-\mathbf{z}_{k}(n)\right\|_*  \leqslant \frac{2 \xi}{\beta(1-\lambda_2)} \log (N \sqrt{K})+3 \xi, 
    \end{equation}
    where $\xi = \sqrt{\Omega^{2}+\frac{\beta^2\max_{n}\text{MSE}(n)}{K}}$ and  $\lambda_2$ is the second largest eigenvalue of the edge-weight matrix $\mathbf{P}$.}
\end{lemma}
\begin{proof}
    See Appendix \ref{App_est_err}.
\end{proof}
One can observe from the above result that dual-variable deviation grows with the numbers of rounds and users following $O(\log(N\sqrt{K}))$. Nevertheless,  it is shown in the next section that with a properly chosen step size, the  noisy primal variables, $\{\hat{\mathbf{x}}_k(N)\}$, can  asymptotically approach the optimal  point as the number of rounds, $N$, increases.  

\subsection{Convergence Analysis}
The convergence of the distributed-optimization algorithm (see Appendix~\ref{App_prelimarary}) as implemented using distributed AirComp is analyzed as follows. To begin with, given  $N$ rounds,  the convergence is evaluated using the metric of expected suboptimality gap,  $\mathsf{E}\left[f\left(\hat{\mathbf{x}}_{k}(N)\right)-f\left(\mathbf{x}^{\star}\right)\right]$,  where $\hat{\mathbf{x}}_{k}(N)$ and $\mathbf{x}^{\star}$ represent the noisy state at device $k$ and the optimal point, respectively. 

\begin{theorem}[Convergence with Distributed AirComp]\label{theo:1} \emph{Given the optimal point $\mathbf{x}^{\star}$, assume that the proximal function satisfying  \eqref{eq:psi} in Appendix~\ref{App_prelimarary} is bounded by some constant $R$ as $\psi\left(\mathbf{x}^{\star}\right) \leq R^{2}$. Let the step size $\alpha(n)$ be chosen as $\alpha(n) = \frac{R \sqrt{1-\lambda_2}}{4 \xi \sqrt{n}}$.  Then the expected suboptimality gap can be bounded for the  ZF multicast  beamforming as 
\begin{equation}\label{eq:conv}
\begin{aligned}
\mathsf{E}\left[ f\left(\hat{\mathbf{x}}_{k}(N)\right)-f\left(\mathbf{x}^{\star}\right)\right]\leqslant \frac{20 R \log (N \sqrt{K})}{\beta\sqrt{N} \sqrt{1-\lambda_2}}\left(\Omega^{2}+\frac{\beta^2\max_{n}\text{MSE}(n)}{K}\right)^{\frac{1}{2}},
\end{aligned}
\end{equation}
and for the MMSE design as 
\begin{equation}\label{eq:conv_MMSE}
\begin{aligned}
    &\mathsf{E}\left[ f\left(\hat{\mathbf{x}}_{k}(N)\right)\!-\!f\left(\mathbf{x}^{\star}\right)\right]
    &\leqslant \frac{20 R \log (N \sqrt{K})}{\beta\sqrt{N} \sqrt{1\!-\!\lambda_2}}\left(\Omega^{2}\!+\!\frac{\beta^2\max_{n}\text{MSE}(n)}{K}\right)^{\frac{1}{2}} \!+\! \frac{\|\mathbf{x}^{\star}\|}{N}\sum_{n=1}^N \sqrt{\frac{\text{MSE}(n)}{K}},
\end{aligned}
\end{equation}
where $\lambda_2$ is the second largest eigenvalue of the edge-weight matrix $\mathbf{P}$.}
\end{theorem}
\begin{proof}
    See Appendix \ref{App_theo_1_2}.
\end{proof}
In the above results, those terms comprising the AirComp error, $\text{MSE}(n)$, reflect the effects of distributed AirComp. For a sanity check, substituting  $\text{MSE}(n)=0$ into the results in Theorem \ref{theo:1} yields the existing result for the ideal case of noiseless channels \cite{duchi_dual_2012}. Next, comparing with \eqref{eq:conv}, the last term at the right-hand size of \eqref{eq:conv_MMSE} reveals slower convergence due to the MMSE beamforming as opposed to the ZF counterpart due to the former's bias in subgradient estimation as shown in the preceding subsection. More critically, the said term does not vanish as the number of rounds, $N$, grows, but instead converges to a constant, $\|\mathbf{x}^{\star}\|\cdot \mathsf{E}[ \sqrt{\text{MSE}(n)/K}]$. As a result, in terms of expected suboptimality gap,  ZF beamforming can significantly  outperform the MMSE counterpart in the regime of low-to-meidum SNRs. Last, the AirComp-error term, introduced by $\frac{\beta^2\max_{n}\text{MSE}(n)}{K}$, scales with the number of rounds, $N$, and devices, $K$, following  $O(\log (N\sqrt{K})/\sqrt{N})$. However, it decays with the increasing transmit SNR following  $O(1/\sqrt{\text{SNR}})$. This leads to a narrowing performance gap between the ZF and MMSE beamforming as the SNR grows.

\section{Simulation Results}
A distributed system with  a varying number of devices, $K$, is simulated. In the system, all channel gains are modelled as i.i.d. Rician fading with the power ratio between the direct and scatter paths being $0.6$ and unit total power. Transmit SNR is defined as $\text{SNR} = P_0/\sigma^2$ and is set equal for all devices. The bandwidth is $B = 1$ MHz. Other case-dependent simulation settings are specified in the sequel. For  distributed  optimization, we consider the specific scenario of  distributed FEEL  that trains a classifier model for  handwritten-digit recognition using the well-known MNIST dataset. This dataset contains $60,000$ labeled training data samples in total. To distribute these samples in a \emph{non i.i.d.} manner, they are first sorted by their digit label, then divided into $20$ shards of size $3,000$, with $2$ shards allocated to each of the $10$ devices. The classifier model is implemented using a 6-layer convolutional neural network (CNN) that consists of two $5\times 5$ convolution layers with ReLU activation, each followed by a $2 \times 2$ max pooling layer, a fully-connected layer with 512 units and ReLU activation, and a final softmax output layer. The total number of parameters is $21,840$. The test accuracy is defined as the lowest test accuracy among all devices. In this scenario, the LSI of each device is a locally computed stochastic gradient for model updating. 

Two benchmarking schemes are described below
\begin{itemize} 
    \item \textbf{Digital communication}: Each coefficient of a transmitted state variable is quantized into $Q = 16$ bits, which are transmitted reliably at a capacity-achieving rate. In each round, devices take turn to broadcast their LSI to peers based on \emph{time-division multiple access} (TDMA). ZF precoding is applied. 
    
    \item \textbf{Single-aggregation AirComp}: In each round, the $K$ AirComp processes in the system are orthogonalized or equivalently  completed sequentially using TDMA, giving the name of the scheme. Consequently, the per-round latency is $K$-time higher than that of the proposed distributed AirComp although the AirComp error is smaller as demonstrated in the sequel. The traditional  beamforming design for a multiple-input-single-output channel is applied together with effective channel inversion for receive-signal alignment, which is a special case of the AirComp design in \cite{zhu_mimo_2019}. 
\end{itemize}

\begin{figure}[t!]
    \centering
    \includegraphics[width=.55\linewidth]{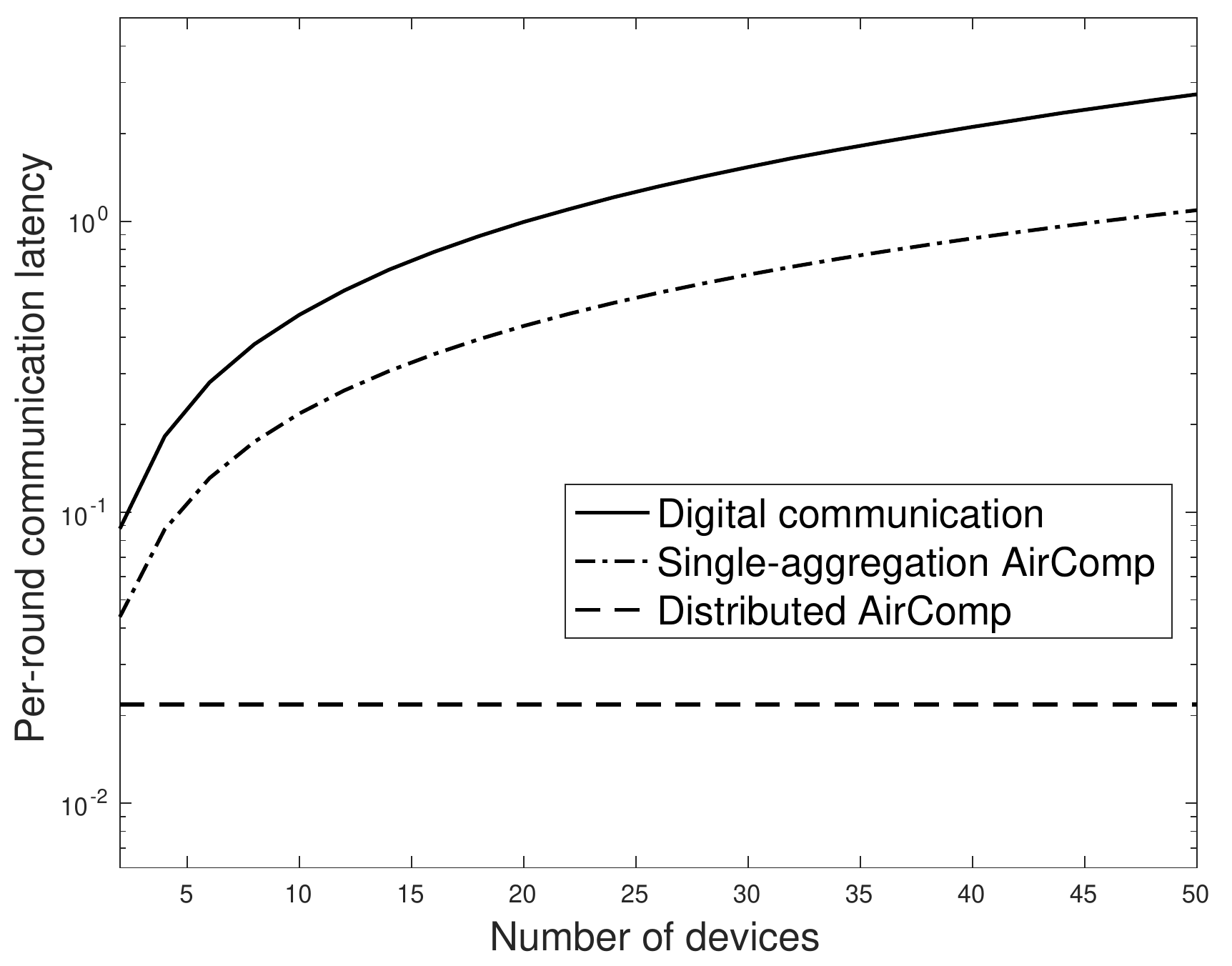}
    \caption{Per-round communication latency with transmit SNR = 20 dB.}
    \label{Fig: dis_t}
\end{figure}

\subsection{Performance of Distributed AirComp}
In Fig. \ref{Fig: dis_t},  the per-round  communication latency for LSI  exchange between $K$ devices is compared between the proposed distributed AirComp and benchmarking schemes for a varying number of devices,  $K$. The transmit SNR  is   $20$ dB. To support the number of devices as many as $50$, each device is provisioned with a large-scale array with $N_t = 100$. One can observe that the feature of simultaneous multicasting enables the proposed distributed AirComp to keep the latency constant instead of increasing with the number of devices as for the benchmarking schemes. Distributed AirComp is observed to achieve much lower latency than the latter with the gap increasing rapidly as $K$ grows. When there are many devices (e.g., $50$), the latency reduction of distributed AirComp is more than \emph{two-order of magnitude} with respect to (w.r.t.)  digital communication and about $50$-time w.r.t. single-aggregation AriComp. 

\begin{figure}[t!]
    \centering
    \subfigure[]
    {
        \includegraphics[width=.47\linewidth]{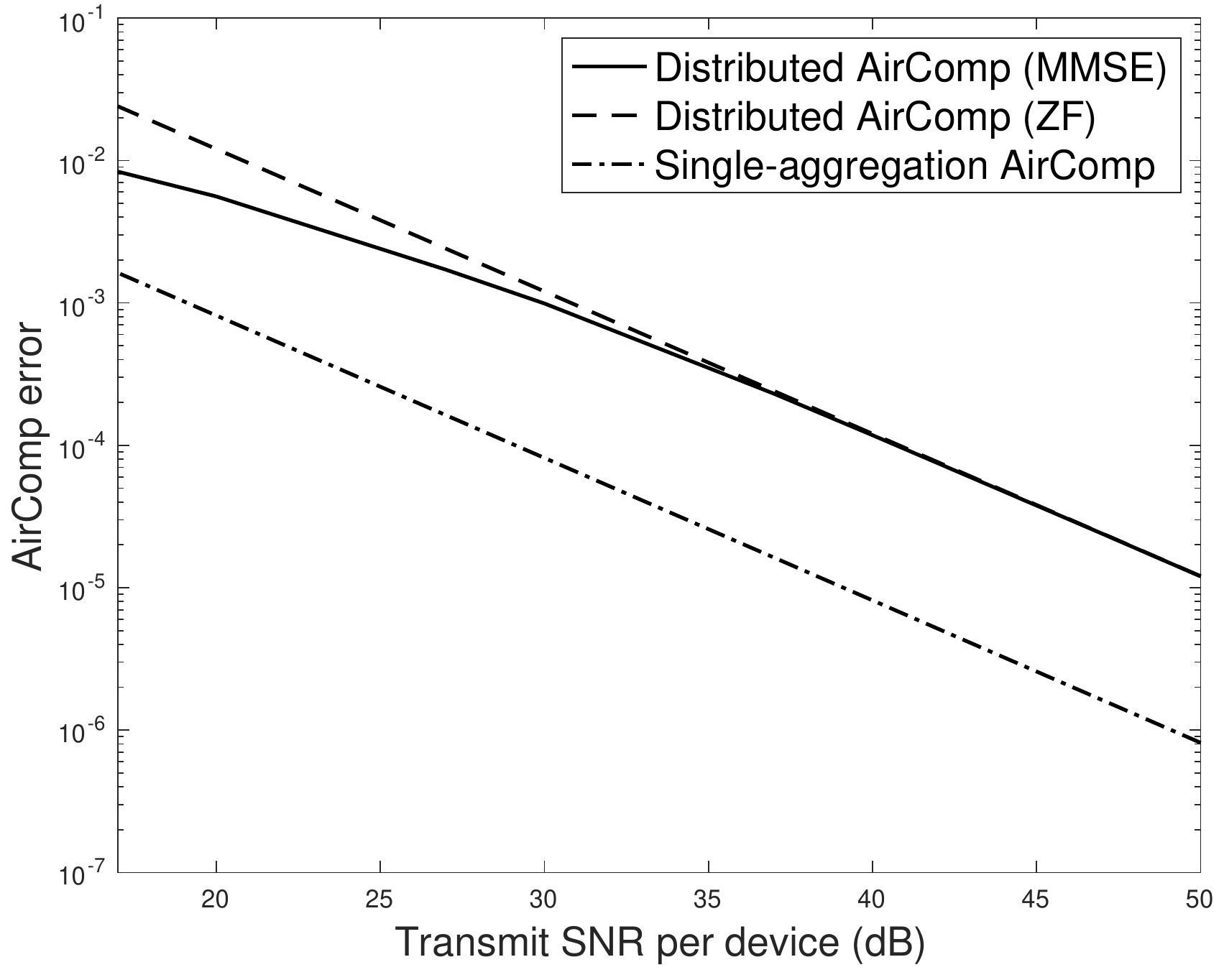}
        \label{Fig:MSE_SNR}
    }\hspace{-3mm}
    \subfigure[]
    {
        \includegraphics[width=.47\linewidth]{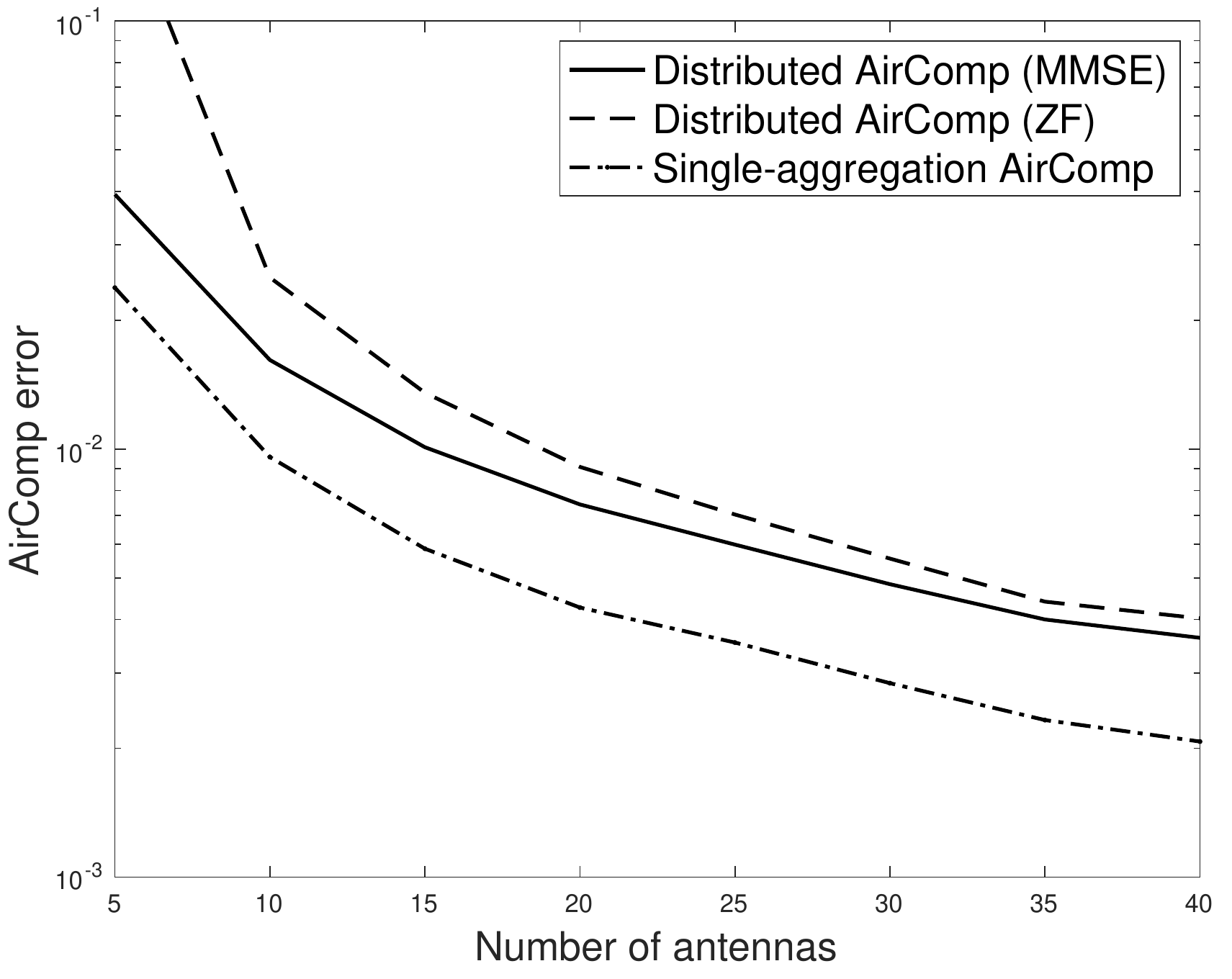}
        \label{Fig:MSE_Nt}
    }\\\vspace{-2mm}
    \subfigure[]
    {
        \includegraphics[width=.47 \linewidth]{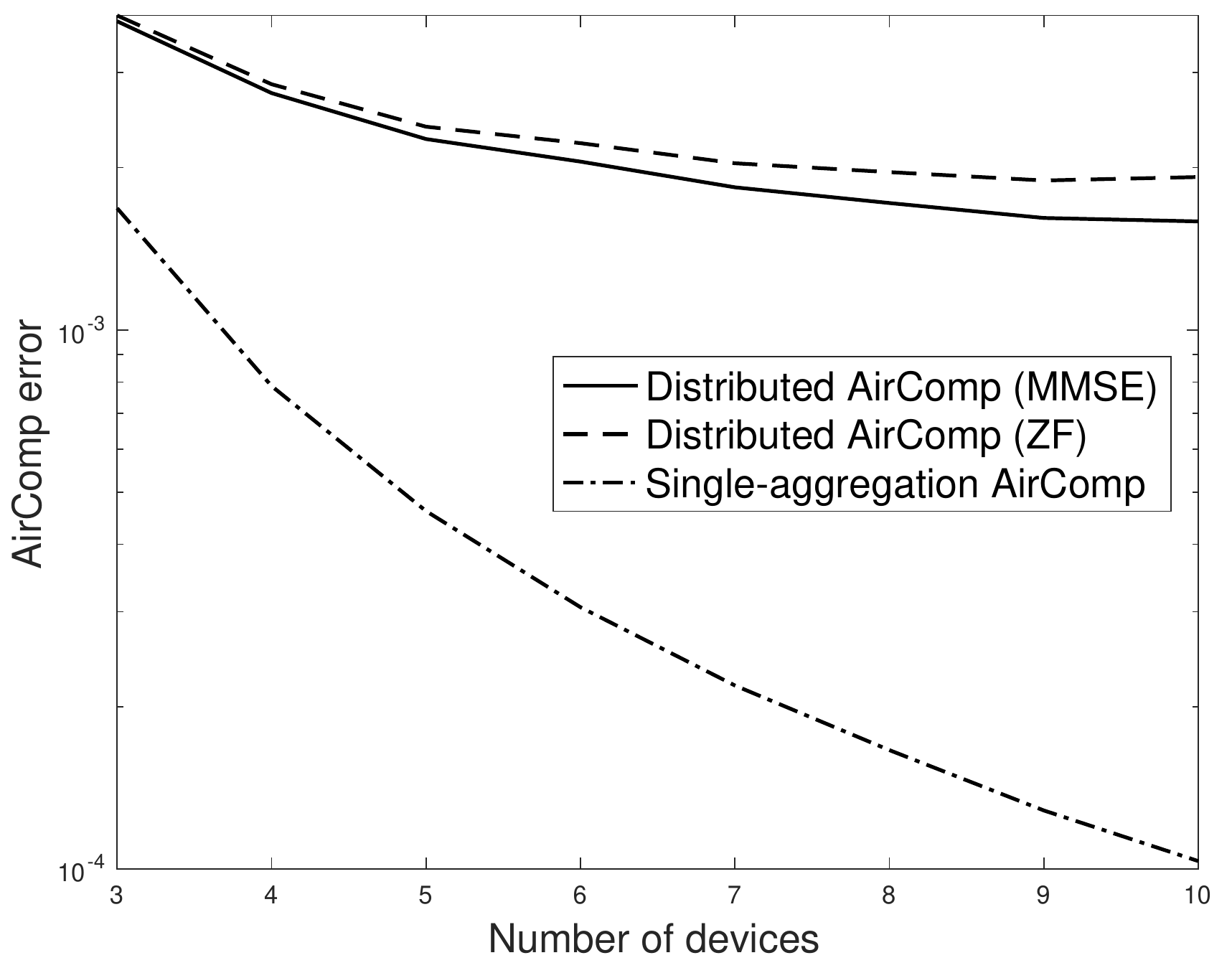}
        \label{Fig:MSE_K}
    }
    \caption{Comparison of AirComp error between distributed AirComp with MMSE and ZF beamforming and single-aggregation AirComp for (a) varying transmit SNR, (b) a varying number of transmit antennas $N_t$, and (c) a varying number of devices $K$. The default value of $K=5$ is used. }
     \label{Fig:net_para}
\end{figure}

Next, the curves of AirComp error versus transmit SNR, number of transmit antennas $N_t$, and number of devices $K$ are plotted in Fig.~ \ref{Fig:MSE_SNR}, \ref{Fig:MSE_Nt}, and  \ref{Fig:MSE_K}, respectively. By default, the number of devices   $K=5$. All simulation results demonstrate decreasing AirComp error as any  of the three parameters increases. Relatively small  array sizes, $N_t = 4$ and $N_t = 18 $, are considered in Fig.~\ref{Fig:MSE_SNR} and  and  \ref{Fig:MSE_K}, respectively, to investigate the limitations of distributed AirComp. One can observe from Fig.~ \ref{Fig:MSE_SNR} that  distributed AirComp, which strives to support $K$ AirComp processes simultaneously despite a small array size, incurs about $10$-time larger AirComp error than single-aggregation AirComp. Furthermore, as shown in Fig.~\ref{Fig:MSE_K}, the former sees that the AirComp error saturates as $K$ increases, indicating the cancellation of the opposite effects of aggregation gain in error suppression and severer receive-signal misalignment. The above disadvantage of distributed AirComp as a price for dramatic latency reduction is alleviated when the array size increases as shown in Fig.~\ref{Fig:MSE_Nt}. For instance, for $N_t = 40$, the error ratio between distributed AirComp and single-aggregation counterpart reduces to below $4$ times. In addition, it can be observed that for distributed AirComp, MMSE multicast beamforming outperforms the ZF design in the regimes of  low-to-medium  SNRs or array size or as the number of devices grows. 

\begin{figure}[t!]
    \centering
    \subfigure[SNR = $10$ dB]{\includegraphics[width=.95\linewidth]{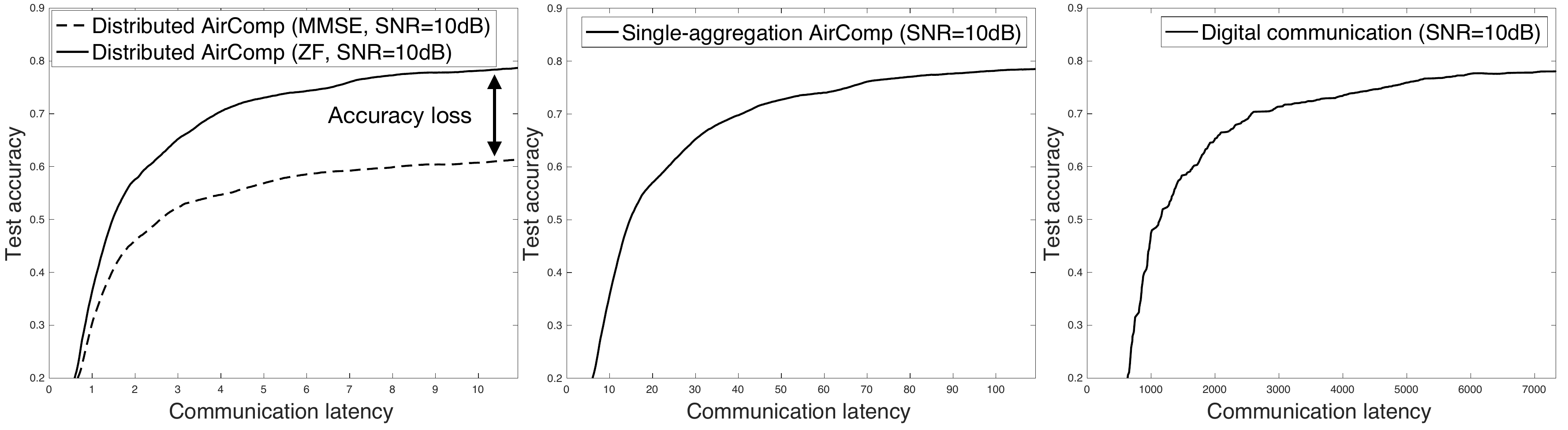}}
    \subfigure[SNR = $20$ dB]{\includegraphics[width=.95\linewidth]{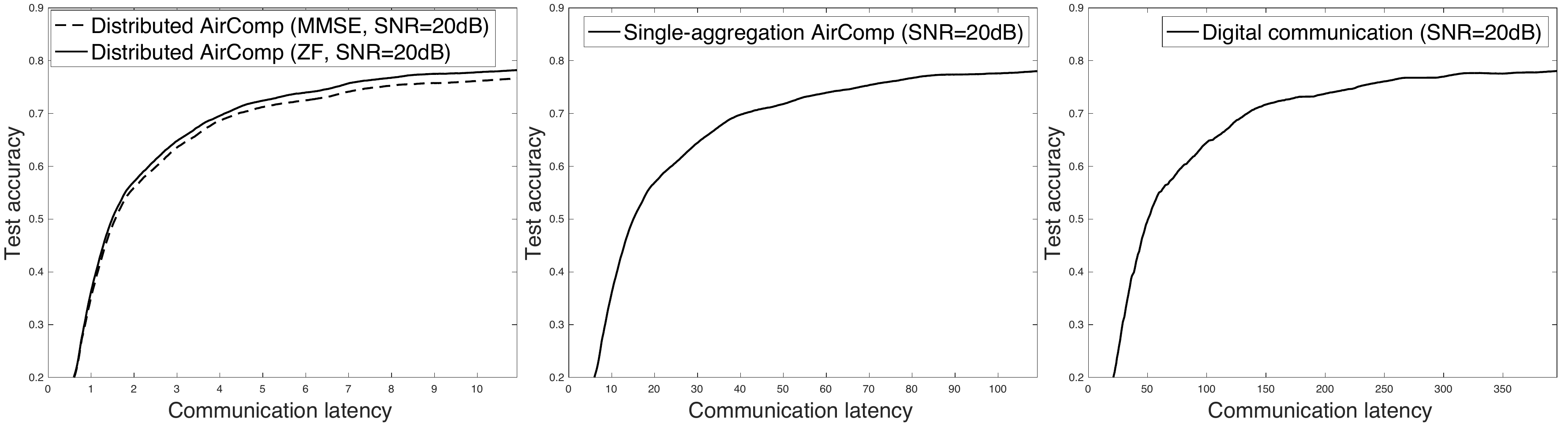}}
    \caption{Test-accuracy comparison between (from left to right)  distributed AirComp, single-aggregation AirComp, and digital communication  for varying communication latency for (a) a medium SNR (SNR = $10$ dB) or (b) a high SNR (SNR = $20$ dB). }
    \label{Fig: dis_opt1}
\end{figure}

\subsection{Performance of  Distributed Optimization}
Consider FEEL, a typical scenario  of distributed optimization, where the number of devices $K = 10$ and each is equipped with $N_t = 18$ antennas. The distributed FEEL algorithm is implemented using either distributed AirComp or one of the benchmarking schemes. Define the communication latency given a number of rounds as the accumulated latency from the start of the task to the current round. The test accuracy of the considered schemes are compared in Fig.~\ref{Fig: dis_opt1} as a function  of communication latency. First, by comparing sub-figures in the same row, the dramatic convergence-speed acceleration achieved by distributed AirComp is aligned with the latency comparison in  Fig.~\ref{Fig: dis_t}. Second, in terms of converged test accuracy,  distributed AirComp with ZF multicast beamforming performs similarly as the benchmarking schemes. Third, in the context of distributed AirComp, MMSE beamforming design is observed from the left-most subfigure of  Fig.~\ref{Fig: dis_opt1}(a) to suffer significant  accuracy loss (i.e., $15\%$)  w.r.t. the ZF counterpart  due to the bias in subgradient estimation as discussed in Section \ref{sec:DO}. The loss is significant in the regime of  low-to-medium SNRs  but varnishes at high SNRs as observed from the left-most subfigure of Fig.~\ref{Fig: dis_opt1}(b).

\section{Concluding Remarks}
To overcome the communication bottleneck in the deployment  of distributed optimization in a distributed network, we have proposed the framework of distributed AirComp that realizes a one-step distributed aggregation of the local states of all devices. The framework  features the seamless integration of simultaneous multicast beamforming of all devices and their full-duplex communication, which makes it possible to support multiple concurrent over-the-air aggregation processes at devices. This has  led to dramatic communication-latency reduction when there are many devices and thus provides a promising solution to data intensive applications such as distributed machine learning and high-mobility applications, such as drone swarm or vehicle platooning. 

This work points to a number of directions warranting follow-up investigations. It is interesting to extend the current single-stream transmission to distributed AirComp with spatial multiplexing, which can further shorten communication latency. On the other hand, the performance of distributed AirComp can be enhanced by distributed resource allocation such as adaptive power control and broadband transmission. Last, particularization of the current design to the area of distributed machine learning provides abundance of cross-disciplinary research opportunities, for example, distributed reinforcement learning with distributed AirComp.

\appendix
\subsection{Distributed Optimization Algorithm}\label{App_prelimarary}
Consider distributed optimization in a distributed network described by an undirected graph denoted as $\mathbb{G} = (\mathbb{V}, \mathbb{E})$, where $\mathbb{V}$ represents  the set of $K$ nodes (e.g., edge devices) and $\mathbb{E}$ represents the set of edges. Each edge, say the one connecting the $k$th and $\ell$th node, is assigned a non-negative weight denoted as $P_{k \ell}$. Then the graph structure can be specified by the $K\times K$ weight matrix $\mathbf{P}$ with $(k, \ell)$th element being $P_{k \ell}$. In particular, two nodes $k\neq \ell$ are connected (i.e., $(k,\ell) \in \mathbb{E}$) if and only if $P_{k \ell}>0$ or otherwise $P_{k \ell} = 0$. The matrix $\mathbf{P}$ is constrained to be doubly stochastic, namely that $\sum_{\ell = 1}^K P_{k \ell }= \sum_{k = 1}^K P_{k \ell}=1$. 

Given the distributed network, the problem of distributed optimization in \eqref{eq:global_obj} can be solved using the classic iterative algorithm of \emph{distributed dual averaging} described as follows \cite{duchi_dual_2012, saha2021decentralized}. To begin with, let each device, say device $k$,  maintain  a two-variable state, $(\mathbf{z}_{k}(n), \mathbf{x}_{k}(n))$, each being a $D\times 1$ row vector with real elements and $\mathbf{x}_{k}(n)\in\mathcal{X}$, which is the support of the objective function. Moreover, a key component of the algorithm is a proximal function $\psi:~\mathbb{R}^D\rightarrow \mathbb{R}$ that is assumed to be 1-strongly convex with respect to some norm $\|\cdot\|$: 
\begin{equation}\label{eq:psi}
    \psi(\mathbf{y}) \geq \psi(\mathbf{x})+\langle\nabla \psi(\mathbf{x}), \mathbf{y}-\mathbf{x}\rangle+\frac{1}{2}\|\mathbf{x}-\mathbf{y}\|^{2} \quad\forall~\mathbf{x}, \mathbf{y} \in \mathcal{X}.
\end{equation}
One example of such a function is $\psi(\mathbf{x}) = \frac{1}{2} \|\mathbf{x}\|^2$. For device $k$, the state variable  $\mathbf{x}_k(n)\in\mathcal{X}$ represents the local estimate of the optimal solution for  \eqref{eq:global_obj} at iteration $n$. Recall that $\tilde{\bm{g}}_{k}(n)$ represents the stochastic subgradient of the local loss function as computed at device $k$ using local data. Then given non-increasing step size $\{\alpha(n)\}$, all devices perform simultaneous updating of their states \cite{saha2021decentralized}: for all $k$ and the $n$ iteration, 
\begin{eqnarray}
\mathbf{z}_{k}(n+1) &=& (1-\beta)\mathbf{z}_{k}(n) + \beta \sum_{\ell \in N(k)} P_{k \ell } \mathbf{z}_{\ell}(n)+\tilde{\bm{g}}_{k}(n), \label{eq:step:1} \\
\mathbf{x}_{k}(n+1) &=&\Pi_{\mathcal{X}}^{\psi}\left(\mathbf{z}_{k}(n+1), \alpha(n)\right),
\end{eqnarray}
where the constant $\beta\in (0, 1)$, $N(k)$ represents the neighbourhood of node $k$ on the graph $\mathbb{G}$, and $\Pi_{\mathcal{X}}^{\psi}$  the projection function defined as
\begin{equation}
    \Pi_{\mathcal{X}}^{\psi}(\mathbf{z}_k(n+1), \alpha(n)):=\underset{\mathbf{x} \in \mathcal{X}}{\operatorname{argmin}}\left\{\langle \mathbf{z}_k(n+1), \mathbf{x}\rangle+\frac{1}{\alpha(n)} \psi(\mathbf{x})\right\} \label{Eq:Project}
\end{equation} 
with $\psi(.)$ being  the proximal function that satisfies \eqref{eq:psi}. For ease of notation, define the weight matrix $\mathbf{W} = (1-\beta)\mathbf{I} + \beta \mathbf{P}$ with the $(k, \ell)$th element denoted as $W_{k \ell}$. Then \eqref{eq:step:1} is rewritten as 
\begin{eqnarray}
\mathbf{z}_{k}(n+1) = \sum_{\ell=1}^K W_{k \ell } \mathbf{z}_{\ell}(n)+\tilde{\bm{g}}_{k}(n), \label{eq:step}. 
\end{eqnarray}
Note that the updating in \eqref{eq:step} is for node $k$ to compute a consensus based on  the states of the peers. On the other hand, the projection function in \eqref{Eq:Project} yields  $\mathbf{x}_k(n+1)$ that minimizes an averaged first-order approximation of  the objective function.
It is assumed that $\psi(\mathbf{x})\geqslant 0~\forall~\mathbf{x}\in\mathcal{X}$ and $\psi(\mathbf{0}) = 0$ without loss of generality.  

The above updating procedure is iterated until a consensus on the minimum of the objective is reached, as the changes of the states of all devices fall below a given threshold. Specifically, given a threshold $\varepsilon$, the convergence criterion can be specified using  the expected suboptimality gap \cite{duchi_dual_2012}:
\begin{eqnarray}\label{eq:metric_c}
    \max_{k}~\mathsf{E}\left[f\left(\hat{\mathbf{x}}_{k}(N)\right)\!-\!f\left(\mathbf{x}^{\star}\right)\right] \leqslant \varepsilon,
\end{eqnarray}
where $\varepsilon$ is a given constant and  the expectation in \eqref{eq:metric_c} taken over all sources of randomness.

\subsection{Proof of Lemma \ref{lemma:ccfp}}\label{App_ccfp}
First, given $\sum_{\ell = 1}^K\sum_{k\neq \ell}^K( \mathbf{h}_{k \ell}^H\mathbf{p}_k+\mathbf{p}_k^H\mathbf{h}_{k \ell})\geqslant 0$, it is direct to see that the objective function of \ref{P1.2} is positive. Denote $x_{k \ell} = \text{Re}(\mathbf{h}_{k \ell}^H\mathbf{p}_k)$, $y_{k \ell} = \text{Im}(\mathbf{h}_{k \ell}^H\mathbf{p}_k)$ as the real and imaginary parts of $\mathbf{h}_{k \ell}^H\mathbf{p}_k$, respectively. Then the numerator and denominator of objective function of Problem \eqref{P1.2} can be written as $A(\mathbf{x}) = 2\sum_{\ell=1}^{K}\sum_{k\neq \ell}^{K} x_{k \ell}$ and $B(\mathbf{x},\mathbf{y}) = 2\sqrt{ K\sigma^{2}+\sum_{\ell=1}^{K}\sum_{k\neq \ell}^{K}(x_{k \ell}^{2}+y_{k \ell}^2)}$ respectively. Notice that the eigenvalues of the Hessian matrix of function $B(\mathbf{x},\mathbf{y})$ : $e_1 \!=\! \frac{2K\sigma^2}{(K\sigma^2+\sum_{k\neq \ell}^{K}(x_{k \ell}^{2}+y_{k \ell}^2))^{3/2}}, e_2\!=\!e_3 \!=\! \cdots \!=\! e_{2K(K-1)} \!=\! \frac{2}{(K\sigma^2+\sum_{k\neq \ell}^{K}(x_{k \ell}^{2}+y_{k \ell}^2))^{1/2}}$  are all positive numbers, hence the $B(\mathbf{x},\mathbf{y})$ is a convex function. Then, since $A(\mathbf{x})$ is a linear function, one can obtain that the the Problem \eqref{P1.2} is a concave-convex fractional program, with its objective function proved to be strictly quasi-concave and pseudo-concave \cite{schaible1976fractional}. Therefore, the objective function of Problem \eqref{P1.2} is a positive unimodal function and Problem \eqref{P1.2} has an unique optimal solution, and the proof is completed. 

\subsection{Proof of Lemma\ref{lemma:mon_proof}}\label{App_mon_proof}
It is first to prove that for any $\alpha_1\leqslant\alpha_2$, $p_k^{\star}(\alpha_1)\leqslant p_k^{\star}(\alpha_2)$ by contradiction. Assume there exists $\alpha_1\leqslant\alpha_2$ such that $p_k^{\star}(\alpha_1)> p_k^{\star}(\alpha_2)$. Since Problem \ref{P1.2} has been proved to be an strictly quasi-concave function, then for $\alpha_1\leqslant\alpha_2$, the size of feasible region $\mathcal{P}_1$ for $\alpha = \alpha_1$ is larger than or equal to the size of $\mathcal{P}_2$ for $\alpha = \alpha_2$, which means we can at least find a set of $\{\mathbf{p}_k\}$ in $\mathcal{P}_1$ that makes $p_k^{\star}(\alpha_1)= p_k^{\star}(\alpha_2)$. This contradicts the previous assumption and thus for any $\alpha_1\leqslant\alpha_2$, $p_k^{\star}(\alpha_1)\leqslant p_k^{\star}(\alpha_2)$.

Moreover, for $\alpha = 0$, we have $\mathbf{p}_k = \mathbf{0}, ~\forall k$, and thus $p_k^{\star}(\alpha)=0$. Then, for any $\alpha > 0$, there must $\exists~\mathbf{p}_k$ such that $\|\mathbf{p}_k\|^2>0$, which means $p_k^{\star}(\alpha)>0$. To this end, $p_k^{\star}(\alpha)$ cannot be a constant function for all $\alpha$. Then the proof is finished.

\subsection{Proof of Lemma \ref{lemma:con_proof}}\label{App_con_proof}
First, the objective of Problem \eqref{P1.4} is linear and the second constraint is a convex set. Then given any $\alpha$, left-hand-side of the first constraint can be regarded as a $2$-norm function and hence convex. Then the right-hand-side of the first constraint is a linear function. Therefore, the proof is completed.

\subsection{Proof of Lemma \ref{lemma:one_full} and Lemma \ref{lemma:direc}}\label{App_two_lemma}
We first give the corresponding Lagrange function of Problem \eqref{P1.4} as below
\begin{equation}
    \begin{aligned}
    &\mathcal{L}(\left\{\mathbf{p}_{k}\right\}, p_{\text{max}}, \lambda, \{v_k\}) = p_{\text{max}} + \\ &\lambda\left(2\alpha\sqrt{(K\sigma^{2}+\sum_{\ell=1}^{K}\sum_{k\neq \ell}^{K}\mathbf{h}_{k \ell }^{H} \mathbf{p}_{k}\mathbf{p}_{k}^{H}\mathbf{h}_{k \ell })}- \left(\sum_{\ell = 1}^K\sum_{k\neq \ell}^K( \mathbf{h}_{k \ell}^H\mathbf{p}_k+\mathbf{p}_k^H\mathbf{h}_{k \ell})\right) \right)+ \sum_{k = 1}^K v_k(\mathbf{p}_k^H\mathbf{p}_k-p_{\text{max}})
    \end{aligned}
\end{equation}
where $\lambda$ and $\{v_k\}$ are the Lagrangian multipliers. The KKT conditions are given by
\begin{equation}\label{eq:kkt2.3}
    \left\{
    \begin{aligned}
        &
        \frac{\partial\mathcal{L}(\left\{\mathbf{p}_{k}\right\},p_{\text{max}}, \lambda, \{v_k\})}{\partial p_{\text{max}}} = 1-\sum_{k=1}^K v_k = 0,\\
        &\frac{\partial\mathcal{L}(\left\{\mathbf{p}_{k}\right\},p_{\text{max}}, \lambda, \{v_k\})}{\partial \mathbf{p}_k} = \lambda\left(\frac{\alpha\left(\sum\limits_{\ell\neq k}^K \mathbf{h}_{k \ell}\mathbf{h}_{k \ell}^H\right)\mathbf{p}_k}{\sqrt{ K\sigma^{2}+\sum_{\ell=1}^{K}\sum_{k\neq \ell}^{K}\left\|\mathbf{h}_{k \ell }^{H} \mathbf{p}_{k} \right\|^{2}}}-\sum\limits_{\ell\neq k}^K \mathbf{h}_{k \ell}\right)+2v_k\mathbf{p}_k=\mathbf{0}, \forall k,\\
        & \lambda\left(\alpha\sqrt{ K\sigma^{2}+\sum_{\ell=1}^{K}\sum_{k\neq \ell}^{K}\left\|\mathbf{h}_{k \ell }^{H} \mathbf{p}_{k} \right\|^{2}}- \left(\sum_{\ell = 1}^K\sum_{k\neq \ell}^K( \mathbf{h}_{k \ell}^H\mathbf{p}_k+\mathbf{p}_k^H\mathbf{h}_{k \ell})\right) \right) = 0, \\
        & v_k(\mathbf{p}_k^H\mathbf{p}_k-p_{\text{max}}) = 0,~\forall k.
    \end{aligned}
    \right.
\end{equation}

In \eqref{eq:kkt2.3}, the first condition indicates $\exists k$, $v_k\neq 0$. This reveals that for any $\alpha$, at least one device transmits with power $p^{\star}(\alpha)$. Then by Lemma \ref{lemma:mon_proof}, one can conclude that for the maximum aligned fraction $\alpha$, at least one device transmits with the maximum power $P_0$.

Next, together with the second and the forth conditions, one can have $\lambda \neq 0$. By the second condition, let a constant $c_0 = \frac{\alpha}{\sqrt{ K\sigma^{2}+\sum_{\ell=1}^{K}\sum_{k\neq \ell}^{K}\left\|\mathbf{h}_{k \ell }^{H} \mathbf{p}_{k} \right\|^{2}}}$, we have
\begin{equation}\label{eq:bj1}
    \left[c_0\lambda\left(\mathbf{H}_k\mathbf{H}_k^H\right)+2v_k\mathbf{I}\right]\mathbf{p}_k = \mathbf{H}_k\bm{{1}}_{(K-1)}.
\end{equation}
By denoting $\mu_k = \frac{2v_k}{c_0\lambda}$, then the desired results are obtained.

\subsection{Proof of Proposition \ref{prop:suff_knt}}\label{App_suff_knt}
The problem of finding the smallest norm solution for $\left\|\mathbf{H}_k^H\mathbf{p}_k-\sqrt{\eta}\mathbf{1}_{(K-1)}\right\|^2 = 0$ is equivalent to the following optimization problem:
\begin{equation}
    \begin{aligned}
    \min_{\mathbf{p}_{k}} 
    & \quad \|\mathbf{p}_k\|^2 \\
    \text { s.t. }
    & \quad \mathbf{H}_k^H\mathbf{p}_k=\sqrt{\eta}\mathbf{1}_{(K-1)}.
    \end{aligned}
    \nonumber
\end{equation}
This problem is convex since both the objective and feasible region are convex. By reusing Lagrange multiplier $\mathbf{\lambda}$, the Lagrange function can be written as
\begin{equation}
    \mathcal{L}(\mathbf{p}_k, \mathbf{\lambda})  = \|\mathbf{p}_k\|^2+\mathbf{\lambda} \left(\mathbf{H}_k^H\mathbf{p}_k-\sqrt{\eta}\mathbf{1}_{(K-1)}\right).
\end{equation}
The KKT conditions are given by
\begin{equation}
    \left\{
    \begin{aligned}
        & \frac{\partial\mathcal{L}(\mathbf{p}_{k},\lambda)}{\partial \mathbf{p}_{k}} = 2\mathbf{p}_k+\mathbf{H}_k\mathbf{\lambda}^H = \mathbf{0},\\
        & \left(\mathbf{H}_k^H\mathbf{p}_k-\sqrt{\eta}\mathbf{1}_{(K-1)}\right) = \mathbf{0}.
    \end{aligned}
    \right.
\end{equation}
Take the first condition into the second one, one can get the expression of $\mathbf{\lambda}$. Then take $\mathbf{\lambda}$ back to the first condition, the desired result is obtained.

\subsection{Proof of Lemma \ref{lemma:est_err}}\label{App_est_err}
From \eqref{eq:air_upd}, the average state in round $n$ is given as
\begin{equation}\label{eq:bar_zt}
\begin{aligned}
    & \overline{\mathbf{z}}(n+1) =\frac{1}{K} \sum_{k = 1}^K \mathbf{z}_{k}(n+1) = \frac{1}{K} \sum_{k = 1}^K \left(\sum_{\ell = 1}^K W_{k \ell }(n) \mathbf{z}_{\ell}(n)+ \hat{\bm{g}}_{k}(n)\right),\\
    & = \overline{\mathbf{z}}(n)+\frac{1}{K} \sum_{k = 1}^K \hat{\bm{g}}_{k}(n).
\end{aligned}
\end{equation}
Define the matrix $\mathbf{\Phi}(n, s)=\mathbf{W}^{n-s+1}$. Then, for the state update $\mathbf{z}_k(n+1)$ at device $k$, it is expanded as follows
\begin{equation}\label{eq:zt_plus_1}
    \mathbf{z}_{k}(n+1)=\sum_{\ell=1}^{K}[\mathbf{\Phi}(n, s)]_{k \ell } \mathbf{z}_{\ell}(s)+\sum_{r=s+1}^{n}\left(\sum_{\ell=1}^{K}[\mathbf{\Phi}(n, r)]_{k \ell } \hat{\bm{g}}_{\ell}(r-1)\right)+\hat{\bm{g}}_{k}(n),
\end{equation}
where $[\mathbf{\Phi}(n, s)]_{k \ell }$ is the $\ell$-th entry of the $k$-th column of $\mathbf{\Phi}(n, s)$. Since the initial state $\mathbf{z}_k(0) = 0$, using \eqref{eq:bar_zt} and \eqref{eq:zt_plus_1} yields
\begin{equation}
    \bar{\mathbf{z}}(n)\!-\!\mathbf{z}_{k}(n)\!=\!\sum_{s\!=\!1}^{n-1} \sum_{\ell\!=\!1}^{K}\left(\frac{1}{K}\!-\![\mathbf{\Phi}(n\!-\!1, s)]_{k \ell }\right) \hat{\bm{g}}_{\ell}(s\!-\!1)+\left(\frac{1}{K} \sum_{\ell\!=\! 1}^{K}\left(\hat{\bm{g}}_{\ell}(n\!-\!1)\!-\!\hat{\bm{g}}_{k}(n\!-\! 1)\right)\right).
\end{equation}
Based on \eqref{eq:var_sub_g}, $\xi^2 = \Omega^2+\frac{\beta^2\max_{n}\text{MSE}(n)}{K}$ denotes an upper bound on the second moment of $\hat{\bm{g}}_{k}(n)$. Then by Jensen's inequality, one has $\left(\mathsf{E}\left[\left\|\hat{\bm{g}}_{k}(n)\right\|\right]\right)^{2} \leqslant \mathsf{E}\left[\left\|\hat{\bm{g}}_{k}(n)\right\|^{2}\right] \leqslant \xi^{2}$ for all $k$. Hence
\begin{equation}\label{exp_cons_1}
    \mathsf{E} \left\|\overline{\mathbf{z}}(n)-\mathbf{z}_{k}(n)\right\|_{*} \leqslant \sum_{s = 1}^{n-1} \xi\left\|\left[\mathbf{\Phi}(n-1, s)\right]_k -\frac{\mathbf{1}}{K}\right\|_{1}+2 \xi.
\end{equation}
Following steps similar to \cite{saha2021decentralized}, we separate the sum in \eqref{exp_cons_1} into two terms by a cutoff point $\hat{n}=n-\frac{\log N \sqrt{K}}{\beta\log \lambda^{-1}}$, where $\lambda_2 = \max \left\{\lambda_{2}(\mathbf{P}),-\lambda_{K}(\mathbf{P})\right\}$ is the second-largest magnitude of eigenvalues of $\mathbf{P}$,
\begin{equation}\label{exp_cons_2}
    \mathsf{E} \left\|\overline{\mathbf{z}}(n)-\mathbf{z}_{k}(n)\right\|_{*} \leqslant \sum_{s = 1}^{\hat{n}} \xi\left\|\left[\mathbf{\Phi}(n-1, s)\right]_k -\frac{\mathbf{1}}{K}\right\|_{1} + \sum_{s = \hat{n}+1}^{n-1} \xi\left\|\left[\mathbf{\Phi}(n-1, s)\right]_k -\frac{\mathbf{1}}{K}\right\|_{1} 
    + 2 \xi.
\end{equation} 
Then for $s\leqslant \hat{n}$,we have $\left\|[\mathbf{\Phi}(n-1, s)]_{k}-\mathbf{1} /K \right\|_{1}\leqslant 1/K$ and for larger $s$, we relax this term via a more loose bound $\left\|[\mathbf{\Phi}(n-1, s)]_{k}-\mathbf{1} /K \right\|_{1}\leqslant 2$. Taking these two bounds into \eqref{exp_cons_2} together with the fact $\log \lambda_2^{-1}\geqslant 1-\lambda_2$, one can obtain the desired result.

\subsection{Proof of Theorem \ref{theo:1}}\label{App_theo_1_2}
The proof is close to the convergence proof for distributed dual averaging \cite{saha2021decentralized} with some modifications. First, an useful lemma is introduced as follows. 
\begin{lemma}\emph{\cite[lemma 8]{saha2021decentralized} Let $\{\mathbf{x}_k(n)\}$ and $\{\mathbf{z}_k(n)\}$ be the primal and dual variables, then the expected suboptimality gap for both ZF and MMSE beamforming cases is bounded as 
\begin{equation} \label{eq:sub_gap}
\begin{aligned}
    &\mathsf{E} \left[ f_{k}\left(\hat{\mathbf{x}}_{k}(N)\right)-f\left(\mathbf{x}^{\star}\right) \right] \leqslant 
    \frac{1}{N \alpha(N)} \psi\left(\mathbf{x}^{\star}\right)+\frac{\xi^{2}}{2 N} \sum_{n=1}^N \alpha(n-1) \\
    &+\frac{L+\xi}{N K} \sum_{n=1}^N \sum_{k = 1}^K \alpha(n) \mathsf{E}\left[\left\|\bar{\mathbf{z}}(n)-\mathbf{z}_{k}(n)\right\|\right]
    +\frac{L}{N} \sum_{n=1}^N \alpha(n) \mathsf{E}\left[\left\|\bar{\mathbf{z}}(n)-\mathbf{z}_{k}(n)\right\|\right] \\
    &+ \mathsf{E}\left[
        \frac{1}{NK} \sum_{n = 1}^N \sum_{k=1}^K\left\langle\bm{g}_{k}(n)-\hat{\bm{g}}_{k}(n), \mathbf{x}_{k}(n)-\mathbf{x}^{\star}\right\rangle,
    \right]
\end{aligned}
\end{equation}
where $\xi= \sqrt{\Omega^2+\frac{\beta^2\max_{n}\text{MSE}(n)}{K}}$, and $\|\cdot\|$ represents the $\ell_2$-norm that is its own dual.
} 
\end{lemma}

For the ZF beamforming case, by \eqref{eq:exp_sub_g_zf}, one can find that the last term of \eqref{eq:sub_gap} is equal to zero. By taking the upper bound in Lemma \ref{lemma:est_err} into \eqref{eq:sub_gap} and with $\psi\left(\mathbf{x}^{\star}\right) \leq R^{2}$, $\alpha(n) = \frac{R \sqrt{1-\lambda_2}}{4 \xi \sqrt{n}}$, the convergence for ZF beamforming in Theorem \ref{theo:1} can be obtained. However, for the MMSE beamforming case, by \eqref{eq:exp_sub_g_mmse}, the last term of \eqref{eq:sub_gap} is nonzero due to the biased gradient. Nevertheless, this term can be bounded as 
\begin{equation}
\begin{aligned}
    \mathsf{E}\left[
        \frac{1}{NK} \sum_{n=1}^N \sum_{k=1}^K\left\langle\bm{g}_{k}(n)-\hat{\bm{g}}_{k}(n), \mathbf{x}_{k}(n)-\mathbf{x}^{\star}\right\rangle\right]
    &\leqslant
    \mathsf{E}\left[
        \frac{1}{NK} \sum_{n = 1}^N \sum_{k = 1}^K\left\|\bm{g}_{k}(n)-\hat{\bm{g}}_{k}(n)\right\| \left\|\mathbf{x}_{k}(n)-\mathbf{x}^{\star}\right\|\right]\\
    &\leqslant
        \frac{1}{NK} \sum_{n = 1}^N \sum_{k = 1}^K \mathsf{E}\left[\left\|\mathbf{\Delta}_k(n)\right\| \right] \mathsf{E}\left[\left\|\mathbf{x}_{k}(n)-\mathbf{x}^{\star}\right\|\right] \\
    &\leqslant
    \frac{\left\|\mathbf{x}^{\star}\right\|}{N}\sum_{n=1}^N \sqrt{\text{MSE}(n)/K}.
\end{aligned}
\end{equation}
The last inequality follows from $\mathsf{E}\left[\left\|\mathbf{x}_{k}(n)-\mathbf{x}^{\star}\right\|\right]\leqslant \mathsf{E}\left[\left\|\mathbf{x}_{k}(0)-\mathbf{x}^{\star}\right\|\right]$ and $\mathbf{x}_{k}(0) = 0, ~\forall k$. Taking this into \eqref{eq:sub_gap}, then Theorem \ref{theo:1} is proved.

\bibliographystyle{IEEEtran}

\end{document}